\documentclass[12pt]{article}

\usepackage{graphicx}
\usepackage{geometry}
\usepackage{hyperref}
\usepackage{amsmath,amssymb,esint,amsthm,bbm,textgreek,mathrsfs,mathtools}
\usepackage{amsfonts}
\usepackage{setspace}
\usepackage{etoolbox}
\usepackage{algpseudocode}
\usepackage{algorithm}
\usepackage{caption}
\usepackage{subfigure}
\usepackage{graphicx}
\usepackage[normalem]{ulem}
\usepackage[table]{xcolor}
\usepackage{booktabs}

\DeclareFontFamily{U}{mathx}{\hyphenchar\font45}
\DeclareFontShape{U}{mathx}{m}{n}{
      <5> <6> <7> <8> <9> <10>
      <10.95> <12> <14.4> <17.28> <20.74> <24.88>
      mathx10
      }{}
\DeclareSymbolFont{mathx}{U}{mathx}{m}{n}
\DeclareFontSubstitution{U}{mathx}{m}{n}
\DeclareMathAccent{\widecheck}{0}{mathx}{"71}

\renewcommand{\vec}[1]{\boldsymbol{#1}}
\renewcommand{\v}{\vec{v}}
\newcommand{\LL}{\mathcal{L}}
\newcommand{\F}{\mathcal{F}}
\newcommand{\C}{\mathcal{C}}

\newcommand{\K}{\mathcal{K}}
\newcommand{\I}{\mathcal{I}}
\newcommand{\J}{\mathcal{J}}
\newcommand{\X}{\mathcal{X}}
\newcommand{\nue}{{\nu}}
\newcommand{\nud}{\widecheck{\nu}}
\newcommand{\grad}{\nabla}

\newcommand\p[2]{\frac{\partial #1}{\partial #2}}
\newcommand\pp[2]{\frac{\partial^2 #1}{\partial #2^2}}
\newcommand\ppp[2]{\frac{\partial^3 #1}{\partial #2^3}}
\newcommand\pppp[2]{\frac{\partial^4 #1}{\partial #2^4}}
\newcommand\STATE[1]{ s }
\newcommand\STATEin[1]{ s }
\newcommand\what[1]{\widehat{#1}}

\DeclareMathOperator*{\argmin}{arg\,min}
\newtheorem{theorem}{Theorem}
\newcommand{\Id}{\operatorname{Id}}

\def\tS{\widetilde{S}}

\def\tv{\widetilde{v}}
\def\tw{\widetilde{w}}
\def\tp{\widetilde{p}}
\def\tbv{\widetilde{\v}}
\def\tu{\widetilde{u}}

\def\Dpartial#1#2{ {\partial #1 \over \partial #2} }
\def\Dpartialmix#1#2#3{ {\partial^2 #1 \over \partial #2\,\partial #3} }

\def\Bmp#1{ \begin{minipage}{#1} }
\def\Emp{ \end{minipage} }
\def\Bmpc#1{ \begin{minipage}[c]{#1} }
\def\Bmpt#1{ \begin{minipage}[t]{#1} }
\def\Bmpb#1{ \begin{minipage}[b]{#1} }

\def\KK{{\mathbb{K}}}
\def\RR{{\mathbb{R}}}

\def\x{{\mathbf x}}

\def\bnabla{\boldsymbol{\nabla}}

\definecolor{Gray}{gray}{0.9}

\newcommand{\revtt}[1]{{\color{red}#1}}

\newcommand\strike[2]{\bgroup\markoverwith
{\textcolor{red}{\rule[.5ex]{2pt}{0.4pt}}}\ULon{#1} \revtt{#2}}

\begin{document}
\title{Optimal Closures in a Simple Model for Turbulent Flows}
 
\author{Pritpal Matharu and Bartosz Protas
\\ \\ 
Department of Mathematics and Statistics, McMaster University \\
Hamilton, Ontario, L8S 4K1, Canada}
\date{\today}

\maketitle

\begin{abstract}
  In this work we introduce a computational framework for determining
  optimal closures of the eddy-viscosity type for Large-Eddy
  Simulations (LES) of a broad class of PDE models, such as the
  Navier-Stokes equation. This problem is cast in terms of
  PDE-constrained optimization where an error functional representing
  the misfit between the target and predicted observations is
  minimized with respect to the {functional form of the eddy viscosity
    in the closure relation.} Since this leads to a PDE optimization
  problem with a nonstandard structure, the solution is obtained
  computationally with a flexible and efficient gradient approach
  relying on a combination of modified adjoint-based analysis and
  Sobolev gradients. By formulating this problem in the continuous
  setting we are able to determine the optimal closure relations in a
  very general form subject only to some minimal assumptions. The
  proposed framework is thoroughly tested on a model problem involving
  the LES of the 1D Kuramoto-Sivashinsky equation, where optimal
  {forms of the eddy viscosity} are obtained as generalizations of the
  standard Smagorinsky model. It is demonstrated that while the
  solution trajectories corresponding to the DNS and LES still diverge
  exponentially, with {such optimal eddy viscosities} the rate of
  divergence is significantly reduced as compared to the Smagorinsky
  model. By systematically finding {optimal forms of the eddy
    viscosity within a certain general class of closure} models, this
  framework can thus provide insights about the fundamental
  performance limitations of these models.
\end{abstract}

\begin{flushleft}
\end{flushleft}



\section{Introduction and Problem Statement} 
\label{sec:Introduction}

Turbulent flows at high Reynolds numbers continue to challenge both
scientists studying their fundamental properties and engineers
interested in diverse technical applications involving fluid
mechanics. In particular, accurate and efficient numerical simulation
of turbulent flows will for the foreseeable future remain an open
problem in computational science, despite advances in both algorithms
and computer architectures. This is because the solutions of the
three-dimensional (3D) Navier-Stokes equation, which is the main
mathematical model assumed to describe the motion of viscous
incompressible fluids, are chaotic and exhibit extreme spatio-temporal
complexity at Reynolds numbers characterizing developed turbulence.
With the Reynolds number defined as $Re = U L /\nu_N$, where $U$
  and $L$ are the characteristic velocity and length scale, and
$\nu_N$ is the coefficient of the kinematic viscosity of the fluid,
{a simple dimensional argument leads} to the conclusion that the
number of discrete degrees of freedom, {e.g., Fourier modes,}
necessary to resolve {a statistically isotropic and homogeneous
turbulent flow down to the smallest active length scales} $\eta$
scales as $\mathcal{O}(Re^{(9/4)})$ \cite{frisch1995turbulence}. This
hints at fundamental limitations on the largest Reynolds numbers for
which direct numerical simulation (DNS) can be performed on the 3D
Navier-Stokes system.

An approach which allows one to get around the aforementioned
difficulty and obtain approximate solutions of the flow problem in
some practical situations relies on the so-called Large-Eddy
Simulation (LES) in which one solves a suitably filtered version of
the governing equations. To define this approach, we consider {the
  linear filtering operation $\widetilde{\v}(\x) := (G_{\delta} \star
  \v)(\x) = \int_{\Omega} G_{\delta}(||\x - \x'||) \v(\x') \, dx'$,
  $\x \in \Omega$, where $\v = [v_1, v_2, v_3]^T$ and $\tbv = [\tv_1,
  \tv_2, \tv_3]^T$ represent the {original and filtered velocity
    fields}, defined in terms of some convolution kernel $ G_{\delta}
  \; : \; \RR^+ \rightarrow \RR^+$ in which $\delta > 0$ denotes the
  cut-off length scale} ({the symbol ``$:=$'' means ``equal to by
  definition''}).
{The LES formulation system is then obtained by applying this
  filtering operation to the Navier-Stokes system and takes the form}
\begin{subequations} 
\label{eq:LES}
\begin{align}
\Dpartial{\tv_i}{t} + \tbv\cdot\bnabla \tv_i
& = - \frac{1}{\rho} \Dpartial{\tp}{x_i} + {\nu_N} \Dpartialmix{\tv_i}{x_j}{x_j} + 
\Dpartial{}{x_j} M_{ij}, \quad i,j=1,2,3,
\label{eq:LESa} \\
\Dpartial{\tv_i}{x_i} &= 0,
\label{eq:LESb}
\end{align}
\end{subequations}
where {$\tp$ is the filtered pressure field and $\rho$ denotes
  the constant density (here and below, Einstein's convention is used
  with repeated indices implying summation; in addition, for brevity,
  we omit here the specification of the flow domain $\Omega$ together
  with the initial and boundary conditions which are assumed generic).
  The quantity} $M_{ij} := \tv_i \tv_j - \widetilde{v_i v_j}$,
$i,j=1,2,3$, is by analogy with the dissipative term already present
in the Navier-Stokes system referred to as the ``subgrid-scale'' (SGS)
stresses \cite{davidson2015turbulence}.
System \eqref{eq:LES} describes {the} evolution of the filtered
(large-scale) velocity field $\tbv$ and, evidently, is not closed
because the SGS stresses depend on the original (unfiltered) velocity
field $\v$.  Since the filtering operation defined by $G_{\delta}$
typically involves attenuation of velocity components with length
scales {smaller} than $\delta$, the SGS stresses thus represent the
averaged effect of these neglected motions on the evolution of the
resolved flow field $\tbv$.  In order to close system \eqref{eq:LES}
one {therefore} needs to represent the SGS stresses in terms of the
resolved field $\tbv$ in some way, which constitutes the celebrated
``turbulence closure problem'' \cite{pope2000turbulent}.

There is a very large body of results concerning the closure problem
{formulated} in different flow conditions, especially in the
engineering literature. Even a brief survey of these results would be
outside the scope of the present study and we refer the reader to the
monographs \cite{pope2000turbulent,Sagaut2006} for more information.
Arguably, the most commonly used family of closure models is of the
eddy-viscosity type in which the SGS stresses are expressed as
\cite{smagorinsky1963general}
\begin{equation}
M_{ij} = {\nu} \tS_{ij}, \quad \text{where} \quad 
\tS_{ij} {:=} \frac{1}{2}\left(\Dpartial{\tv_i}{x_j} + \Dpartial{\tv_j}{x_i} \right),
\label{eq:Sm}
\end{equation}
in which ${\nu}$ is the eddy viscosity (to simplify the notation used
below and in contrast to the {commonly employed} convention, we choose
to adopt a simple symbol for the eddy viscosity and put a subscript on
the {kinematic} viscosity.  The approaches to determining this
quantity can be classified as {\em algebraic}, in which some algebraic
relation is {postulated} between the filtered field $\tbv$ and the
eddy viscosity ${\nu}$ (such as the celebrated Smagorinsky model
discussed below), and {\em differential}, in which the eddy viscosity
is {assumed to depend on} some additional quantities whose transport
is described by suitable partial-differential equations (PDEs), {such
  as,} e.g., the fluctuating kinetic energy $k$ and dissipation
$\epsilon$ in the $k-\epsilon$ model {often used as a closure in
  Reynolds-averaged Navier-Stokes (RANS) equations}.  Since the SGS
stresses are assumed {in \eqref{eq:Sm}} to depend on the strain field
$\tS_{ij}$, the eddy-viscosity models have a similar structure to the
dissipative term {$\nu_N \Delta \v$} already present in the
Navier-Stokes, so its inclusion {in the equation} has the effect of
changing the coefficient of this term {from $\nu_N$} to $({\nu_N} +
{\nu})$.  We note however that, unlike the {kinematic} (molecular)
viscosity {coefficient} ${\nu_N}$ which is constant, the eddy
viscosity ${\nu}$ depends on the resolved field $\tbv$ and therefore
introduces an additional nonlinearity. {In addition to the classical
  Smagorinsky model,} there exist many other approaches to approximate
the eddy viscosity, including dynamic Smagorinsky models
\cite{lilly1992proposed} relying on Germano's commutator identity
\cite{germano_1992} and the structure-function models
\cite{Lesieur1996}, to mention just {two}.  For a survey of recent
progress in the field of turbulence modelling we refer the reader to
\cite{Durbin2018}.
Regardless of details, in all cases these closure models are
postulated based on empirical grounds, albeit usually with a strong
physical justification, with a small number of parameters {requiring}
calibration from data.

Since most closure models are derived in {a} heuristic manner, such
approaches to turbulence modelling have been sometimes criticized as
lacking scientific rigor and therefore difficult to generalize to
flows different from the ones for which they have been calibrated. The
objective of the present study is to provide insights about how well
{eddy-viscosity} {closure} models can in principle perform.
{This is achieved by finding, via solution of a suitable optimization
  problem, a} mathematically optimal form of the eddy viscosity for a
given flow.  More precisely, while in algebraic closure models a
simple relationship is typically postulated for the dependence of the
eddy viscosity ${\nu}$ on the resolved flow field $\tbv$ involving a
small number of free parameters, in our proposed approach we will
determine the functional form of this dependence optimally in a very
general setting subject only to some minimal assumptions.

To fix attention, we will consider what is arguably the most common
algebraic closure model, namely, the Smagorinsky model postulating the
eddy viscosity in the form ${\nu} = C_s^2 \delta^2 \, \left(2 \tS_{ij}
  \tS_{ij} \right)^{1/2}$ in which $C_s$ is an adjustable constant
known as the Smagorinsky coefficient \cite{smagorinsky1963general}.
Although the Smagorinsky model is rather simple, it is quite popular
and serves as the ``workhorse'' for many LES computations. It is
known, however, to possess certain {deficiencies} 
{such as assuming} the eddy viscosity to be zero if the resolved
strain $\tS$ vanishes and the fact that the eddy viscosity is positive
{otherwise}, implying that the closure is strictly dissipative
\cite{rodi2013large}.  In {our} {study} we introduce a
computational framework {for determining} an optimal Smagorinsky
model in which the eddy viscosity is allowed to have a very general
functional dependence on the magnitude of the resolved strain field
$\left(2 \tS_{ij} \tS_{ij} \right)^{1/2}$ found by matching the
predictions of the LES model against a given ``target'' field (e.g.,
obtained by solving the original Navier-Stokes problem via DNS or from
an experiment). Since this eddy viscosity is {\em optimal} within the
class of Smagorinsky-type models, its properties will provide insights
about how well this class of models can in principle perform.

There have been earlier attempts to determine turbulence closure
models {with some optimality properties}.  Langford \& Moser
\cite{langford1999optimal} and then Das \& Moser \cite{das2002optimal}
developed an approach for isotropic turbulence in which {motions
  at} subgrid scales were treated as stochastic and the closure was
determined by minimizing the modelling error using stochastic
estimation techniques.  This approach was tested on a range of models,
including a stochastically forced one-dimensional (1D) Burgers
equation and 3D Navier-Stokes system.

An emerging family of approaches uses various machine-learning
techniques such as neural networks to deduce closure models with
certain optimality properties from data. In this context we mention
the investigations
\cite{GamaharaHattori2017,ling_kurzawski_templeton_2016,PARISH2016758},
whereas the state-of-the-art in this field is discussed in the review
papers \cite{duraisamy2018turbulence,jimenez_2018,kutz_2017}.
Data-driven machine-learning methods, in addition to other data-driven
{techniques}, have been utilized for computational prediction,
{modelling}, and diagnosis of various turbulent flows
\cite{Maulik2018,PanDuraisamy2018,Iliescu2018ROM}. We note that, while
the approach developed in the present study can also be classified as
``data-based'', it does not rely on machine learning, but rather on
the calculus of variations and rigorous methods of PDE-constrained
optimization. More specifically, recognizing that closure models are
in fact forms of constitutive relations, we {extend} the method
developed in \cite{bukshtynov2013optimal,bukshtynov2011optimal} to
infer optimal constitutive relations from data. In the context of
hydrodynamics such techniques have already been used to tackle the
simpler problem of finding optimal closures for finite-dimensional
reduced-order models in \cite{pnm14,protas_noack_osth_2015} and in
vortex dynamics \cite{dp15a}. Applications of this approach in the
field of electrochemistry are discussed in
\cite{sethurajan2015accurate}.

Since {our goal here} is to provide a ``proof of the concept''
for {the proposed} approach by introducing and validating it in a
general context, we shall focus on a 1D model problem which will allow
us to avoid the technical complexities inherent in dealing with the 3D
Navier-Stokes system. We will thus consider the Kuramoto-Sivashinsky
equation defined on the periodic domain $\Omega = [0,2\pi]$
\begin{subequations} 
\label{KS}
\begin{align}
&\p{w(t, x)}{t} + {\nu_4} \, \pppp{w(t, x)}{x} + \nu_2 \, \left[\pp{w(t, x)}{x} + w(t, x) \, \p{w(t, x)}{x}\right] = 0,& & (t, x) \in (0, T] \times \Omega,  \label{KS_1}\\
&\frac{\partial^{(i)} w}{\partial x^{(i)}}(t, 0) = \frac{\partial^{(i)} w}{\partial x^{(i)}}(t, 2\pi),& & i = 0, \dotsc, 3, \\
&w(0, x) = w_0(x),
\end{align}
\end{subequations} 
where $T > 0$ is the length of the time window considered, ${\nu_4},
\nu_2 \in \RR^{+}$ are parameters whereas $w_0$ is an appropriate
initial condition. The reason for choosing this particular model
problem is that its solutions are known to exhibit important features
characteristic of actual turbulence governed by the 3D Navier-Stokes
system, namely, chaotic {and multiscale dynamics} with {significant}
spatio-temporal complexity. These properties arise from an interplay
between the linear and nonlinear terms in \eqref{KS_1}: the
second-order negative diffusion term injects energy at {large
  scales} which is then transferred by the nonlinear interactions to
{small scales} where it is eventually dissipated by the
fourth-order dissipative term.  Unlike the Burgers equation, in which
a similar behavior may only arise from the inclusion of a somewhat
artificial stochastic forcing term \cite{bk07}, the
Kuramoto-Sivashinsky equation intrinsically exhibits a more
turbulence-like behavior.  Originally, this equation was proposed as a
model for instabilities {of} interfaces and flame fronts
\cite{sivashinsky1988nonlinear}, and ``phase turbulence'' in chemical
reactions \cite{kuramoto1978diffusion}. Beyond its original purpose,
this equation has been used as a model for hydrodynamic turbulence and
is commonly {employed} as a testbed to study new approaches
\cite{MR2920624}.

The structure of the paper is as follows: in the next section we
introduce the {problem of finding an optimal form of the eddy
  viscosity} in the context of the Kuramoto-Sivashinsky equation
\eqref{KS}; its solution based on a gradient approach is described in
a general context in Section \ref{sec:optimal}, whereas the set-up of
the particular problem considered here is described in Section
\ref{sec:Problem}; details of the numerical approach are presented in
Section \ref{sec:Numerical}; our computational results are discussed
in Section \ref{sec:Results}, whereas a summary and final conclusions
are deferred to Section \ref{sec:final}; {proof of a key
  {result} is provided} in an appendix.

\section{\label{sec:Model_Problem}{Eddy-Viscosity} Closures for
  the Kuramoto-Sivashinsky Equation} 

In this section we formulate an LES system corresponding to the
Kuramoto-Sivashinsky system \eqref{KS} where the closure uses a
Smagorinsky-type eddy-viscosity model of the general form
\eqref{eq:Sm} reduced to 1D.  We must first define the filtering
operation to extract the resolved scales from the solutions and
{for this purpose we} shall use a sharp spectral filter also
employed in \cite{das2002optimal}. It is defined in terms of the
Fourier transform $\widehat{G}_{\delta}(k)$ of its kernel (hats
``\;$\widehat{\cdot}$\;'' will hereafter denote Fourier coefficients)
\begin{equation} 
\label{eq:hG}
\widehat{G}_{\delta}(k) =
\begin{cases}
1, &|k| \leq k_{\text{max}} \\
0, &\text{otherwise}
\end{cases},
\end{equation}
where the {cut-off} length scale $\delta = 2 \pi /
k_{\text{max}}$ and $k_{\max}$ is the maximum resolved wavenumber.
Clearly, \eqref{eq:hG} defines a low-pass filter which removes all
Fourier components with wavenumbers larger than $k_{\text{max}}$.
Applying this filter to \eqref{KS_1}, we obtain the filtered version
of the Kuramoto-Sivashinsky equation
\begin{equation}
\label{eq:KSles}
\p{{\tw}}{t} + {\nu_4} \, \pppp{ \tw}{x} + \nu_2 \, \Big[\pp{\tw}{x} + \frac{1}{2} \p{{(\widetilde{\tw \tw})}}{x}\Big] + \p{M(w)}{x} = 0,
\end{equation}
where $\tw := G_{\delta} \star w$ and the last term represents the
effect of the SGS stresses
\begin{equation} 
\label{eq:Mw}
M(w) := \frac{\nu_2}{2} \left[ \, \widetilde{ww} - \widetilde{\tw\tw} \,  \right].
\end{equation}
{The reason for the sign difference with respect to the
  corresponding expression in \eqref{eq:LESa} is the sign of the
  dissipative term in \eqref{eq:KSles}.}  For clarity, we will
{hereafter} use the symbol $\tu$ to denote the solution of the
LES problem for the Kuramoto-Sivashinsky equation, which should be
contrasted with $\tw$ obtained by filtering the solution $w$ of the
DNS problem {\eqref{KS}}.  Since expression \eqref{eq:Mw} depends
on the original unresolved field $w$, it must be modelled in terms of
$\tw$ in order to close equation \eqref{eq:KSles}. For this purpose we
will use a 1D analogue of the {eddy-viscosity} closure model \eqref{eq:Sm} adapted to
the Kuramoto-Sivashinsky equation, namely,
\begin{equation} 
\label{eq:Mtw}
M = {\nu}(|{s}|) \, \ppp{{\tu}}{x}, \quad \text{where} \quad {s} := \p{{\tu}}{x},
\end{equation}
in which ${\nu(|s|)}$ is the eddy viscosity. The ansatz in
\eqref{eq:Mtw} is chosen such that the order of derivatives in the
resulting model term will match the order (four) of the dissipative
term in the Kuramoto-Sivashinsky system \eqref{KS}. The filtered
Kuramoto-Sivashinsky system \eqref{eq:KSles} equipped with such a
closure model then becomes the LES system with the following form
\begin{subequations}
\label{eq:KSLES}
\begin{align}
&\p{ {{\tu}}}{t} + {\nu_4} \, \pppp{{\tu}}{x} + \nu_2 \Big[\pp{{\tu}}{x} + \frac{1}{2} \p{({{\widetilde{\tu \tu}}})}{x}\Big] + \frac{\partial}{\partial x} \Big[{\nu}(|{s}|)\, \ppp{{\tu}}{x}\Big] = 0, \label{eq:KSLESa} \\
&\frac{\partial^{(i)} {\tu}}{\partial x^{(i)}}(t, 0) = \frac{\partial^{(i)} {\tu}}{\partial x^{(i)}}(t, 2\pi), \quad \quad i = 0, \dotsc, 3, \label{eq:KSLESb} \\
&{\tu}(0, x) = \tw_0(x). \label{eq:KSLESc}
\end{align}
\end{subequations}
{We now introduce two important definitions:}
\begin{itemize}
\item $\I := [\alpha, \beta]$, where $\alpha := \min_{x\in{\Omega}, \
    t \in [0,T]} |{s}|$ and $\beta := \max_{x\in{\Omega}, \ t \in
    [0,T]} |{s}|$, referred to as the ``identifiability interval'', is
  the range of values {attained} by the magnitude of the resolved
  strain $|{s}|$ {in} the solution of the LES problem
  \eqref{eq:KSLES} with the initial data $\tw_0$ over the time
  interval $[0,T]$,

\item $\LL := [a, b]$, where $a < \alpha$ and $b > \beta$, will serve
  as the domain of definition of the function defining the eddy
  viscosity, i.e., ${\nu} \; : \; \LL \rightarrow \RR$; since the
  identifiability interval will in general depend on the initial data
  $\tw_0$ and the length of the time window $T$, i.e., $\I =
  \I(\tw_0,T)$, it is important to choose the domain $\LL$ such that
  it will contain all possible identifiability intervals, that is
  {such that} $\cup_{\tw_0,T} \I(\tw_0,T) \subset \LL$, as this
  will ensure that the eddy viscosity is always defined; in practice,
  it is convenient to adopt a larger domain $\LL$ possibly also
  including points outside any identifiability interval $\I(\tw_0,T)$,
  i.e., $\LL \setminus \cup_{\tw_0,T} \I(\tw_0,T) \neq \varnothing$;
  with this in mind, we shall set $a = \inf_{x\in\Omega, \ t \in
    [0,T]} |s| = 0$ and $b > \sup_{x\in\Omega, \ t \in [0,T]} |s|$.
\end{itemize}

The counterpart of the Smagorinsky model in the present setting will
then take the form
\begin{equation}
{\nu}({s}) =  C_s^2 \delta^2 |{s}|.
\label{eq:nusKS}
\end{equation}
{Our goal now} is to find an optimal form for the eddy viscosity
as a function of the resolved strain ${s}$, ${\nu} = {\nu}(|{s}|)$,
generalizing the Smagorinsky model \eqref{eq:nusKS}. This eddy
viscosity will be {\em optimal} in the sense that the corresponding
solutions of the LES system \eqref{eq:KSLES} will be as close as
possible (in a suitable least-squares sense) to solutions of the
original Kuramoto-Sivashinsky system \eqref{KS} obtained for the same
initial data $w_0$. Formulation and solution of this optimization
problem are presented below.

\section{Optimization Approach to finding Eddy Viscosity}
\label{sec:optimal}

{In this section we first state the optimization problem which
  will be used to determine the optimal form of the eddy viscosity. It
  is formulated here in a very general continuous setting and to solve
  this problem we use a gradient-descent approach in which the key
  element is a suitably-defined gradient representing the sensitivity
  of solutions to the LES system \eqref{eq:KSLES} to modifications of
  the functional form of the eddy viscosity.  Finally, we ensure that
  these gradients are sufficiently smooth such that the resulting
  optimal eddy viscosity will be well defined.}

{Starting from some initial guess $\nu_0$, the optimization
  procedure will iteratively adjust} the eddy viscosity such that the
corresponding solutions ${\tu}$ of the LES problem \eqref{eq:KSLES}
will match as closely as possible the ``true'' evolution governed by
the original Kuramoto-Sivashinsky system \eqref{KS}, i.e., the DNS.
To fix attention, this matching will be determined in terms of $N$
``observations'' made by applying suitably-defined observation
operators $H_{i} \; : \; {H^1}(\Omega) \longrightarrow
\mathbb{R}$, $i = 1, \dots, N$, to the LES and DNS solutions,
${\tu}(t,\cdot)$ and $w(t,\cdot)$, continually for all $t \in [0,T]$.
{The symbol ${H^1}(\Omega)$ denotes the Sobolev space of
  continuous functions with square-integrable gradients \cite{af05}.}
We will use this general formulation to introduce our approach here
and will define specific forms of these observation operators in
Section \ref{sec:Problem} which will then be used in our computational
examples in Section \ref{sec:Results}.  The ``target'' observations
will thus have the form $m_i(t) := H_{i} w(t,\cdot)$, $i = 1, \dots,
N$.

We see that in order for the LES system \eqref{eq:KSLES} to be
satisfied in the classical (strong) sense, the eddy viscosity
${\nu}(|{s}|)$ must possess certain minimum regularity as a function of
${s}$. Due to some technical reasons which will become apparent below,
we must have
\begin{equation} 
\label{eq:nue_reg}
{\nu}(|{s}|) \ \textnormal{piecewise $C^3$ on} \ \mathcal{L}.
\end{equation}
Since optimization problems are most conveniently formulated in
Hilbert spaces \cite{l69}, we will assume the eddy viscosity
${\nu}(|{s}|)$ to be an element of the Sobolev space $H^3(\LL)$ of
functions defined on $\LL$ with square-integrable third derivatives
\cite{af05} (a precise definition of the inner product in this space
will be provided below). The objective functional $\J \; : \; H^3(\LL)
\rightarrow \RR$ will therefore have the form of the least-squares
error between the target observations $\{ m_i(t) \}_{i=1}^N$ and the
corresponding observations of solutions ${\tu}$ of the LES problem
\eqref{eq:KSLES} obtained for the given eddy viscosity ${\nu}$, i.e.,
\begin{equation} 
\label{eq:J}
\J({\nu}) = \frac{1}{2} \int_0^T \sum_{i = 1}^{N} \left[ m_i(t) - H_{i} {\tu}(t, x; \nue) \right]^2 \ dt,
\end{equation}
such that the optimization problem takes the form
\begin{equation} 
\label{eq:minJ}
\nud {:=} \underset{\nue \in H^3(\LL)} {\argmin} \J(\nue),
\end{equation}
where $\nud$ is the optimal eddy viscosity.

To find a local minimizer of \eqref{eq:J}, we shall use a
gradient-based optimization approach in which the optimal eddy
viscosity $\nud$ can be computed iteratively as $\nud = \underset{n
  \to \infty}{\text{lim}} \nu^{(n)}$, where
\begin{align} 
\label{eq:desc}
\begin{cases}
\nu^{(n+1)} &= \nu^{(n)} - \tau^{(n)} \, \grad_{\nue}\J(\nu^{(n)}), \quad \quad n=1, 2, \dots, \\
\nu^{(1)} &= \nu_0,
\end{cases}
\end{align}
in which $\grad_{\nue}\J(\nue)$ is the gradient of the cost functional
\eqref{eq:J} with respect to the eddy viscosity $\nue$, $\tau^{(n)}$
is the step length along the descent direction at the $n$th iteration{,}
and $\nu_0$ is the initial guess for the eddy viscosity. An optimal
step size $\tau^{(n)}$ can be found by solving the following
line-minimization problem \cite{nw00}
\begin{equation}
\tau^{(n)} = \argmin_{\tau > 0} \J(\nu^{(n)} - \tau \, \grad_{\nue}\J(\nu^{(n)})).
\label{eq:taun}
\end{equation}
We add that due to the local nature of this approach, iterations
\eqref{eq:desc} can only produce local minimizers and determining
whether any of them is also a global minimizer is in general not
possible.  A key element of the gradient-descent approach
\eqref{eq:desc} is evaluation of the gradient
$\grad_{\nue}\J(\nu^{(n)})$ and this step is discussed below.

\subsection{Adjoint-Based Gradients}
\label{sec:adjoint}

While adjoint calculus has had a long history in PDE-constrained
optimization starting with \cite{control:lions1}, the optimization
problem defined by \eqref{eq:KSLES}, \eqref{eq:J} and \eqref{eq:minJ}
has in fact a somewhat non-standard structure and therefore requires
special techniques. The reason is that the control variable ${\nu}$ in
\eqref{eq:minJ} is a function of ${s}$, which {itself is} a function
of the {\em dependent} variable ${\tu}$ in \eqref{eq:KSLES},
cf.~\eqref{eq:Mtw}, whereas standard adjoint-based methods allow one
to solve PDE optimization problems in which the control {is a function
  of the {\em independent} variables only} (here, $x$ and $t$).  A
generalization of the adjoint-based approach overcoming this
limitation was developed in
\cite{bukshtynov2013optimal,bukshtynov2011optimal} and {in the present
  study we adopt a variant of this technique with a number of
  modifications. Most importantly, here the eddy viscosity is a
  function of the {\em magnitude of the gradient of} the state
  variable, cf.~\eqref{eq:Mtw}, rather than of the state variable
  itself, which leads to additional steps in the derivation of the
  adjoint sensitivities. Moreover, increased regularity requirements
  imposed on the eddy viscosity, cf.~\eqref{eq:nue_reg}, result in a
  more complicated form of the system defining the Sobolev gradients
  whose solution in turn necessitates a more refined numerical
  approach than used in
  \cite{bukshtynov2013optimal,bukshtynov2011optimal}.  Here we present
  key elements only of our approach} and the reader is referred to
Section \ref{sec:grad} for {a proof of the main result}. We begin by
computing the G\^{a}teaux (directional) differential of the cost
functional \eqref{eq:J} with respect to ${\nu}$
\begin{align} 
\label{eq:dJ}
\mathcal{J}'({\nu}; \nu') :=& \lim_{\epsilon\to0} \frac{\mathcal{J}({\nu} + \epsilon \nu') - \mathcal{J}({\nu})}{\epsilon} = \frac{d}{d\epsilon} \ \mathcal{J}\Big({\nu} + \epsilon \nu'\Big) \Big|_{\epsilon = 0}\nonumber \\
=& \int_0^T \sum_{i = 1}^{N} [H_{i}{\tu}(t, x; {\nu}) - m_i(t)] \, H_{i}{\tu'}(t, x; {\nu}, \nu') \, dt,
\end{align}
where $\nu'$ is a perturbation of the eddy viscosity ${\nu}$ and
${\tu'}(t, x; \nue, \nu')$ satisfies the corresponding linear
perturbation system obtained from \eqref{eq:KSLES}, cf.~relations
\eqref{eq:pert}--\eqref{eq:dKSLES} in Section \ref{sec:grad}.  The
(local) minimizer {$\nud$} of \eqref{eq:minJ} requires the directional
derivative \eqref{eq:dJ} to vanish for all perturbations $\nu'$, i.e.,
$\underset{\nu' \in {H^3}(\LL)}{\forall} \ \J'(\nud; \nu') = 0$.  Away
from the minimizer $\nud$ we can use the G\^{a}teaux differential to
obtain the gradient $\grad_{\nue}\J$ required by the descent algorithm
\eqref{eq:desc}.  To do this, we invoke the Riesz representation
theorem \cite{b77} and the fact that the directional derivative
\eqref{eq:dJ} is a bounded linear functional when viewed as a function
of $\nu'$, to obtain
\begin{equation} 
\label{eq:Riesz}
\mathcal{J}'({\nu}; \nu') = \left\langle \grad_{\nue}\mathcal{J}, \nu' \right\rangle_{\X(\LL)},
\end{equation}
where $\langle \cdot, \cdot \rangle_{\X}$ is an inner product in the
Hilbert space $\X$. As regards the choice of this space, we will
consider $\X = L^2(\LL)$ and $\X = H^3(\LL)$ endowed with the
following inner products 
\begin{subequations}
\label{eq:ip}
\begin{alignat}{2} 
&\forall_{p_1,p_2 \in L^2(\mathcal{L})} & \quad
\left\langle p_1, p_2 \right\rangle_{L^2(\mathcal{L})} 
& = \int_{a}^b p_1 \, p_2  \, ds, \label{eq:ipL2} \\
&\forall_{p_1,p_2 \in H^3(\mathcal{L})} & \quad
\left\langle p_1, p_2 \right\rangle_{H^3(\mathcal{L})} 
& = \int_{a}^b p_1 \, p_2 + \ell_1^2 \, \frac{dp_1}{ds}\frac{dp_2}{ds} + \ell_2^4 \, \frac{d^2p_1}{ds^2}\frac{d^2p_2}{ds^2} + \ell_3^6 \, \frac{d^3p_1}{ds^3}\frac{d^3p_2}{ds^3} \, ds, \label{eq:ipH3}
\end{alignat}
\end{subequations}
where $\ell_1$, $\ell_2$ and $\ell_3$ in \eqref{eq:ipH3} are
``length-scale'' parameters (we note that as long as $0 < \ell_1,
\ell_2, \ell_3 < \infty$, the corresponding inner products are
equivalent, in the precise sense of norm equivalence). {While the
  Sobolev gradient obtained in the space $H^3(\LL)$ must be used in
  computations in \eqref{eq:desc}, it is convenient to first derive
  the gradient defined with respect to the $L^2$ topology.}

We note that relation \eqref{eq:dJ} is not consistent with the Riesz
form \eqref{eq:Riesz}, because the perturbation $\nu'$ does not appear
in it explicitly as a factor, but is instead hidden in the
perturbation equation \eqref{eq:dKSLESa}. {However, as
  demonstrated by the theorem stated below,} relation \eqref{eq:dJ}
can be transformed to the desired Riesz form \eqref{eq:Riesz} in which
{$\X = L^2(\LL)$, allowing us to identify the corresponding
  gradient of the cost functional.
\begin{theorem}
  Suppose $\nu' \in H^3(\LL)$ and $\X = L^2(\LL)$. Then, the G\^ateaux
  differential admits a Riesz representation \eqref{eq:Riesz} in which
  the $L^2$ gradient is given by
\begin{equation} 
\label{eq:gradL2}
\grad_{\nu}^{L^2}\J(s) 
= - \frac{d}{ds}\int_{0}^T \int_{0}^{2\pi} \Xi_{\left[\alpha, \frac{\partial \tu}{\partial x}(t, x)\right]}(s) \, \p{{\tu}^*(t, x)}{x} \, \ppp{{\tu}(t, x)}{x} \, dx \, dt,
\quad s \in \I(\tw_0,T),
\end{equation}
where $\Xi_{[a_1, a_2]}(s)$ is the characteristic function of the
interval $[a_1, a_2] \in \I$ whereas the adjoint state ${\tu}^* \; : \; [0,2\pi] \times
[0,T] \rightarrow \RR$ is the solution of the following
system
\begin{subequations}
\label{eq:aKSLES}
\begin{align}
-\p{{\tu}^*}{t} + {\nu_4} \, \pppp{{\tu}^*}{x} + & \nu_2 \Big[\pp{{\tu}^*}{x} - {\tu}\p{{\tu}^*}{x}\Big] + \p{}{x}\Big[2 \, \frac{d\nu}{{ds}} \, \p{{\tu}}{x} \, \ppp{{\tu}}{x}\, \p{{\tu}^*}{x}\Big] + \ppp{}{x} \Big[\nu \, \p{{\tu}^*}{x} \Big] \nonumber \\
& = \sum_{i = 1}^{N} H_{i}^* \, [H_{i}{\tu}-m_i], \label{eq:aKSLESa} \\
\frac{\partial^{(i)} {\tu}^*}{\partial x^{(i)}}(t, 0) &= \frac{\partial^{(i)} {\tu}^*}{\partial x^{(i)}}(t, 2\pi), \quad \quad i = 0, \dotsc, 3, \label{eq:aKSLESb} \\
{\tu}^*(T, x) &= 0 \label{eq:aKSLESc}
\end{align}
\end{subequations}
in which $H_{i}^*$ are the adjoints of the observation operators
$H_i$, $i=1,\dots,N$. 
\label{thm:gradL2}
\end{theorem}
\begin{proof}
See Section \ref{sec:grad}.
\end{proof}}

We remark that the adjoint system \eqref{eq:aKSLES} is a {\em
  terminal-value} problem and must be therefore integrated backwards
in time whereas its coefficients are determined by the solution
${\tu}(t,x)$ of the (forward) LES problem \eqref{eq:KSLES} around
which linearization is performed (see Section \ref{sec:grad}). When
the adjoint system is properly defined, its solutions contain
information about the sensitivity of the solutions to the LES problem
\eqref{eq:KSLES}, and hence also the error functional \eqref{eq:J}, to
perturbations of the functional form of the eddy viscosity in
\eqref{eq:Mtw}.  From the structure of the last term on the left-hand
side (LHS) in \eqref{eq:aKSLESa} it is also clear that in order for
the adjoint system to be satisfied in the classical (strong) sense,
the eddy viscosity must possess the minimum regularity specified in
\eqref{eq:nue_reg}.

As defined in \eqref{eq:gradL2}, the $L^2$ gradient may not be used in
optimization algorithm \eqref{eq:desc} because it does not possess the
required regularity, cf.~\eqref{eq:nue_reg}, and is defined only on
the identifiability interval $\I$ which in general is smaller than the
domain {of definition} $\LL$ of the eddy viscosity (this latter
issue could in principle be remedied by extending the $L^2$ gradient
\eqref{eq:gradL2} onto $\LL \backslash \I$ with zeros). In order to
get around these difficulties we will derive the corresponding Sobolev
gradients \cite{Neuberger2010book,pbh04} by setting $\X = H^3(\LL)$ in
the Riesz identity \eqref{eq:Riesz} which, upon noting
\eqref{eq:ipH3}, leads to
\begin{align} 
\label{eq:RieszH3}
\mathcal{J}'(\nu; \nu') &= \left\langle \grad_{\nu}^{L^2}\mathcal{J}, \nu' \right\rangle_{L^2(\mathcal{L})} 
= \left\langle \grad_{\nu}^{H^3}\mathcal{J}, \nu' \right\rangle_{H^3(\mathcal{L})} \nonumber \\ 
&= \Big\langle \grad_{\nu}^{H^3}\mathcal{J}, \nu' \Big\rangle_{L^2(\mathcal{L})} 
+ \ell_1^2 \, \Big\langle \frac{d(\grad_{\nu}^{H^3}\mathcal{J})}{ds}, \frac{d\nu'}{ds} \Big\rangle_{L^2(\mathcal{L})}
+ \ell_2^4 \, \Big\langle \frac{d^2(\grad_{\nu}^{H^3}\mathcal{J})}{ds^2}, \frac{d^2\nu'}{ds^2} \Big\rangle_{L^2(\mathcal{L})} \nonumber \\
& \quad\quad+ \ell_3^6 \, \Big\langle \frac{d^3(\grad_{\nu}^{H^3}\mathcal{J})}{ds^3}, \frac{d^3\nu'}{ds^3} \Big\rangle_{L^2(\mathcal{L})}.
\end{align}
Sobolev gradients are determined subject to certain boundary
conditions imposed on $s = a$ and $s = b$
\cite{bukshtynov2013optimal,bukshtynov2011optimal} which in turn
determine the behavior of the corresponding properties of the optimal
eddy viscosity $\nud(s)$ at $s = a,b$. There is some freedom as
regards this choice and we shall use 
\begin{equation} 
\label{eq:gradH3bc}
\frac{d^{(2i+1)} \, (\grad_{\nu}^{H^3}\J)}{ds^{(2i+1)}} \Big|_{s=a} =
\frac{d^{(i)} \, (\grad_{\nu}^{H^3}\J)}{ds^{(i)}} \Big|_{s=b} = 0, \quad \quad i = 0, \dotsc, 2,
\end{equation}
which implies that in the gradient iterations \eqref{eq:desc} the
odd-degree derivatives of the eddy viscosity with respect to $s$ will
remain unchanged with respect to the initial guess $\nu_0$ at $s = a$
and the value of ${\nu}$ and its first two derivatives will remain
unchanged at $s = b$. We remark that, in particular, the gradient
iterations \eqref{eq:desc} are allowed to modify the value of $\nud$
at $s = a$. It should be noted, however, that when the domain $\LL$ is
much larger than the identifiability interval $\I$ , i.e., $a \ll
\alpha$ and $\beta \ll b$, the boundary conditions such as
\eqref{eq:gradH3bc} have little effect on the optimal eddy viscosity
$\nud$. Performing integration by parts with respect to $s$ in
\eqref{eq:RieszH3} required number of times and noting that due to the
judicious choice of the boundary conditions \eqref{eq:gradH3bc} all
boundary terms vanish we obtain
\begin{align*} 
\mathcal{J}'(\nu; \nu') &=  \int_{a}^b \grad_{\nu}^{L^2}\mathcal{J} \, \nu' \, ds  
= \int_{a}^b \Big[\grad_{\nu}^{H^3}\mathcal{J} 
- \ell_1^2 \, \frac{d^2(\grad_{\nu}^{H^3}\mathcal{J})}{ds^2} 
+ \ell_2^4 \, \frac{d^4(\grad_{\nu}^{H^3}\mathcal{J})}{ds^4} 
- \ell_3^6 \, \frac{d^6(\grad_{\nu}^{H^3}\mathcal{J})}{ds^6} \Big] \, \nu' \, ds
\end{align*}
which because of the arbitrariness of the perturbation $\nu' \in
H^3(\LL)$ implies
\begin{equation} 
\label{eq:gradH3}
\left[\Id 
- \ell_1^2 \, \frac{d^2}{ds^2} 
+ \ell_2^4 \, \frac{d^4}{ds^4} 
- \ell_3^6 \, \frac{d^6}{ds^6}\right] \, \grad_{\nu}^{H^3}\J(s) = \grad_{\nu}^{L^2}\mathcal{J}(s),
\qquad s \in \LL.
\end{equation}
In \eqref{eq:gradH3} the $L^2$ gradient appearing on the RHS is
extended from the identifiability interval $\I$ to the entire domain
$\LL$ with zeros. We note that extraction of Sobolev gradients by
solving the boundary-value problem
\eqref{eq:gradH3bc}--\eqref{eq:gradH3} can be interpreted as applying
a low-pass filter to the $L^2$ gradient.  Indeed, the Fourier
transform of \eqref{eq:gradH3} yields
${\left[\what{\grad_{\nu}^{H^3}\J}\right]_k = \F(k) \,
  \left[\what{\grad_{\nu}^{L^2}\mathcal{J}}\right]_k}$, where $\F(k)
:= (1 + \ell_1^2 \, k^2 + \ell_2^4 \, k^4 + \ell_3^6 \, k^6)^{-1}$,
which shows that the cut-off for filtering is determined {by} the
length scales $\ell_1$, $\ell_2$ and $\ell_3$. 
Adjusting these parameters will have a significant
effect on the rate of convergence of gradient iterations
\eqref{eq:desc} \cite{pbh04}.

\section{\label{sec:Problem}Problem Set-up} 

The formulation in Section \ref{sec:optimal} was introduced in a quite
general setting and in this section we first provide concrete
definitions of the observation operators $H_i$, $i=1,\dots,N$, (and
their adjoints $H_i^*$), and then discuss the choice of the various
parameters defining the problem of finding the optimal eddy viscosity
$\nud$, cf.~\eqref{eq:minJ}.

\subsection{Observation Operators}
\label{sec:H}

We start by defining the observation operators and consider two
choices corresponding to observations in the physical and in the
spectral (Fourier) space.

\subsubsection{Physical-Space Observations}
\label{sec:Hph}

We will assume here that for all times $t \in [0,T]$ the solution
${\tu}(t,x)$ of the LES system \eqref{eq:KSLES} is observed at a certain
number of observation points $\{x_i\}_{i=1}^N$ uniformly distributed
over the spatial domain $[0,2\pi]$.  The observation operator $H_i \;
: \; {H^1}(0,2\pi) \rightarrow \RR$ associated with the $i$th point is
then given by (we express it here in an integral form in order to
facilitate obtaining its adjoint $H_i^*$)
\begin{equation*}
H_{i} {\tu}(t,\cdot) := \int_{0}^{2\pi} \delta (x_i - x) \, {\tu}(t,x) \, dx, \qquad i = 1, \dots, N. 
\label{eq:Hph}
\end{equation*}
The adjoint system \eqref{eq:aKSLES} also involves the adjoints $H_i^*
\; : \; \RR \rightarrow {H^{-1}}(0,2\pi)$ of $H_i$,
$i=1,\dots,N$, which can be obtained from the following duality
relations
\begin{align*}
\left\langle f, H_{i} {\tu}(t,\cdot) \right\rangle_{\mathbb{R}} 
&= f \, \int_{0}^{2\pi} \delta (x_i - x) \, {\tu}(t,x) \, dx 
= \int_{0}^{2\pi} \left[ f \, \delta (x_i - x) \right] \, {\tu}(t,x) \, dx \\
&= \left\langle H_{i}^* f, {\tu}(t,\cdot) \right\rangle_{{H^{-1}(\Omega)\times H^1(\Omega)}},
\end{align*}
where $f \in \RR$ and $\langle \cdot, \cdot \rangle_{\RR}$ denotes the
(scalar) product of two real numbers. From this we thus deduce
\begin{equation*}
\forall f \in \RR \qquad H_{i}^*  f := f \, \delta (x_i - x), \qquad i = 1, \dots, N. 
\label{eq:aHph}
\end{equation*}

\subsubsection{Fourier-Space Observations}
\label{sec:Hs}

Since the periodic Kuramoto-Sivashinsky system \eqref{KS} is
employed here as a ``toy model'' for homogeneous turbulence, another
natural way to define the observation operators $H_i$ is in terms of
the Fourier (e.g., cosine) transform of the state and this is the
second choice we shall consider 
\begin{equation}
H_{i} {\tu}(t,\cdot) := \int_{0}^{2\pi} \cos(k_i x) \, {\tu}(t,x) \, dx, 
\qquad k_i \in \KK, \quad  i = 1, \dots, N,
\label{eq:Hs}
\end{equation}
where $\KK$ is the set of wavenumbers corresponding to the observed
Fourier components (with cardinality $|\KK| = N$). The adjoints
$H_i^*$ of these observation operators are then obtained by
considering the duality relations
\begin{equation*}
\left\langle f, H_{i} {\tu}(t,\cdot) \right\rangle_{\mathbb{R}} 
= f \, \int_{0}^{2\pi} \cos(k_i x) \, {\tu}(t,x) \, dx 
= \int_{0}^{2\pi} \left[ f \,\cos(k_i x) \right] \, {\tu}(t,x) \, dx
= \left\langle H_{i}^* f, {\tu}(t,\cdot) \right\rangle_{L^2(\Omega)},
\end{equation*}
from which we deduce
\begin{equation*}
\forall f \in \RR \qquad H_{i}^*  f := f \, \cos(k_i x), \qquad i = 1, \dots, N. 
\label{eq:aHs}
\end{equation*}

\subsection{Physical Parameters}
\label{sec:params}

The long-time behavior of the solutions $w$ of the
Kuramoto-Sivashinsky system \eqref{KS} is determined by the parameters
${\nu_4}$ and $\nu_2$ \cite{MR2920624}. In our study we shall use the
values ${\nu_4} = 1$ and $\nu_2 = 100$ chosen such that, after an
initial transient, the solution $w$ will on average feature 7 waves
(``coherent structures'') present in the domain during the evolution
(this number coincides with the wavenumber $k_0$ of the most
(linearly) unstable mode of the Kuramoto-Sivashinsky system \eqref{KS}
linearized about the zero state $w(t,x) = 0$ \cite{MR2920624}). Given
the chaotic nature of the Kuramoto-Sivashinsky system in this
parameter regime and our interest in the long-time evolution, as the
initial condition $w_0$ we will take a certain state on the turbulent
attractor. When defining the LES system \eqref{eq:KSLES} we will take
the cut-off {wavenumber} in filter \eqref{eq:hG} to be
$k_{\text{max}} = 16$, which falls in the ``inertial range'' not too
far from the wavenumber $k_0 = 7$ characterizing the most unstable
modes (which can be interpreted as ``forcing''), cf.~Figure
\ref{fig:Hi}. In the solution of the optimization problem
\eqref{eq:J}--\eqref{eq:minJ} we will use $N = 8$ observations and in
the case when the observation operators are defined in the Fourier
space, cf.~Section \ref{sec:Hs}, we will consider two sets $\KK$ of
observed wavenumbers
\begin{subequations}
\label{eq:KK}
\begin{align}
& \bullet \ \text{equispaced:} \quad \KK = \{1,3,\dots,15\},   \label{eq:KKa} \\
& \bullet \ \text{clustered around $k_0$:} \quad \KK = \{4,5,\dots,11\}.   \label{eq:KKb} 
\end{align}
\end{subequations}
The number of observations $N$ is chosen to be smaller than the number
$k_{\text{max}}$ of the Fourier components resolved in the LES
\eqref{eq:KSLES}, cf.~Figure \ref{fig:Hi}. As regards the initial
guess $\nu_0$ for the eddy viscosity in the gradient descent
\eqref{eq:desc} we will take the Smagorinsky model \eqref{eq:nusKS}
with $C_s = 0.002$. The domain $\LL = [a,b]$ will have boundaries
$a=0$ and $b=400$ which for the given problem set-up ensures that $b >
\sup_{x\in\Omega, \ t \in [0,T]} |s|$.
\begin{figure}
  \centering
\mbox{
  \subfigure[]
  {
    \includegraphics[scale=0.44]{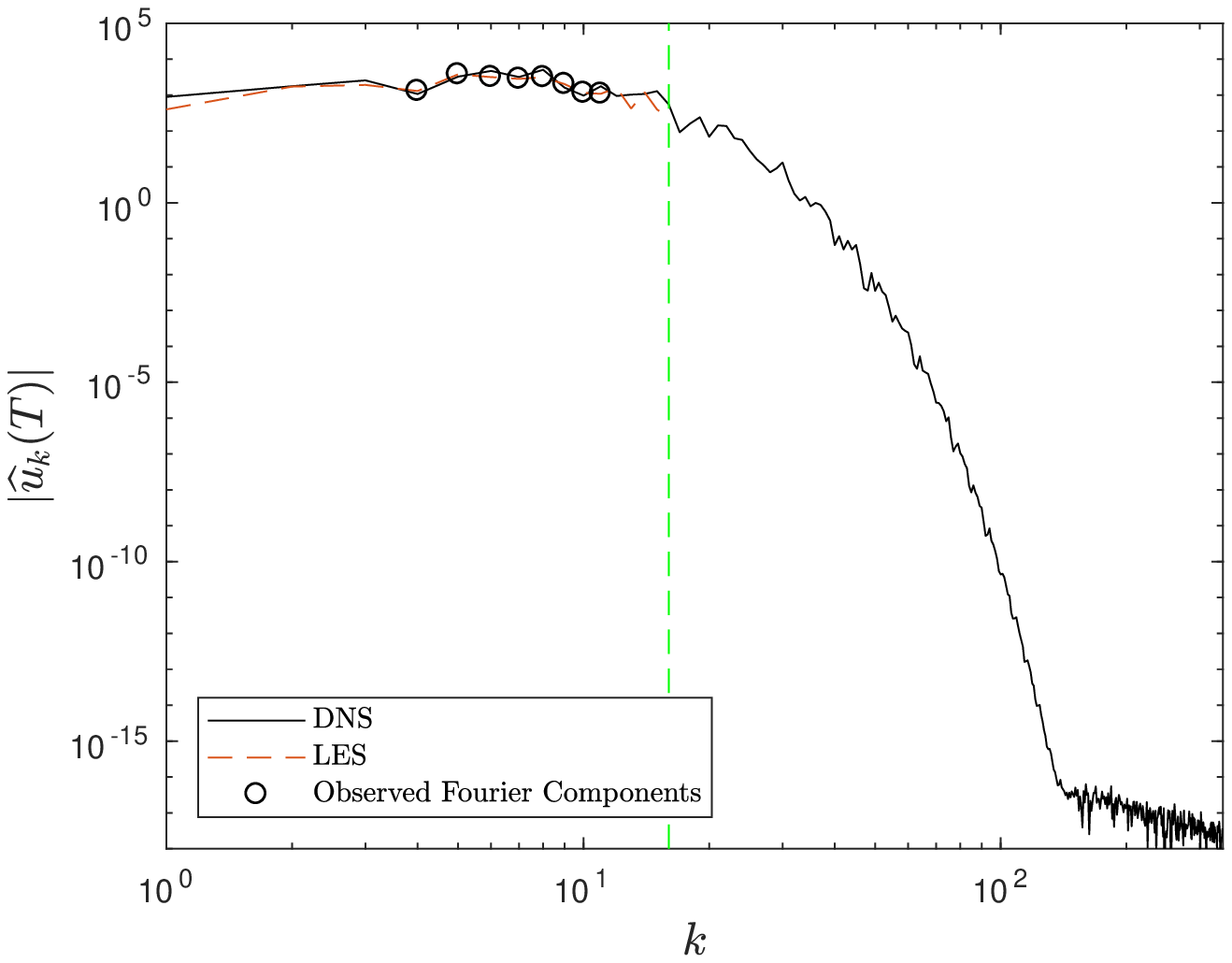}
  }\qquad\qquad
  \subfigure[]
  {
    \includegraphics[scale=0.44]{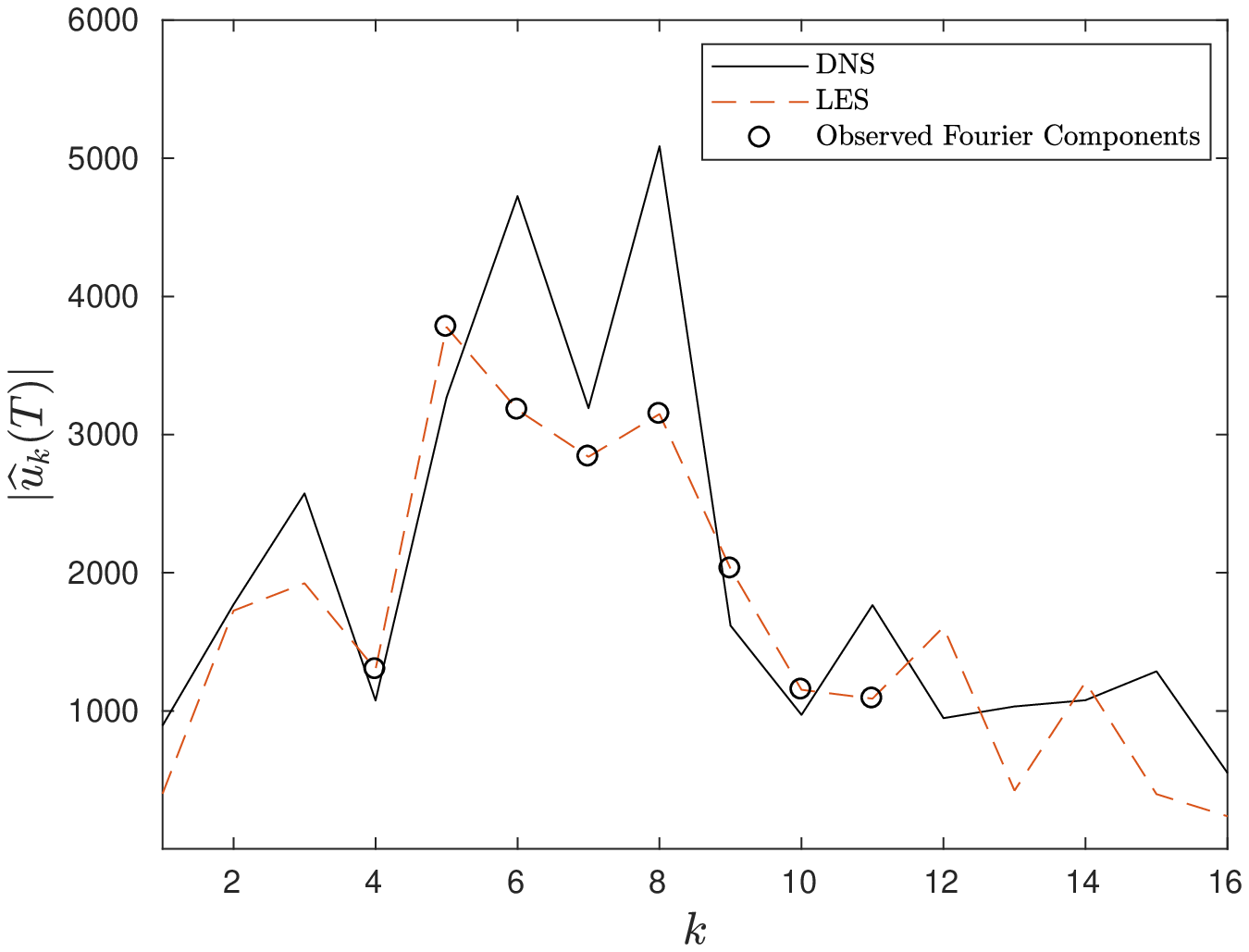}
  }
}
\caption{(a) The spectra of representative solutions of the DNS
  problem (black solid lines) and the LES problem (dashed red lines)
  together with the Fourier components whose real parts are recorded
  by the observation operators $H_i$, $i=1,\dots,8$, defined in
  Section \ref{sec:Hs}, cf.~\eqref{eq:Hs} and \eqref{eq:KKb} (black
  circles).  (b) Magnification of the wavenumber region around $k_0$.
  The green vertical line in panel (a) represents {the cut-off
    wavenumber} $k_{\text{max}}$.  The DNS and LES problems \eqref{KS}
  and \eqref{eq:KSLES} are solved numerically as described in Section
  \ref{sec:disc}.}
\label{fig:Hi}
\end{figure}

\section{\label{sec:Numerical}Numerical Approach} 

The gradient-descent approach \eqref{eq:desc} is formulated in the
continuous (``optimize-then-discretize'') setting \cite{g03} and
evaluation of the gradient expression \eqref{eq:gradL2} requires
solutions of the LES and the adjoint systems \eqref{eq:KSLES} and
\eqref{eq:aKSLES}. In this section we first discuss the numerical
solution of these PDE problems and the computation of the Sobolev
gradients via {\eqref{eq:gradH3bc}--}\eqref{eq:gradH3}. Then we
describe the implementation of the gradient-descent algorithm
\eqref{eq:desc}.

\subsection{Discretization}
\label{sec:disc}

The LES and adjoint systems \eqref{eq:KSLES} and \eqref{eq:aKSLES}
involve model terms with state-dependent eddy viscosity $\nu(|{s}|)$
and in order to represent this expression, in addition to discretizing
the space and time domains $[0,2\pi]$ and $[0,T]$, we also need to
discretize the state domain $\LL$. The former two domains are
discretized using grids with equispaced points and steps sizes $\Delta
x = 2 \pi / N_x$ and $\Delta t$, where $N_x$ is the number of grid
points in space, whereas the state domain $\LL$ is discretized with
$N_s$ Chebyshev points. The original Kuramoto-Sivashinsky system
\eqref{KS}, its LES version \eqref{eq:KSLES}, and the adjoint system
\eqref{eq:aKSLES} are solved using the standard Fourier
pseudo-spectral method \cite{Canuto1993book} where dealiasing based on
the 3/2 rule is performed in the case of the Kuramoto-Sivashinsky
system \eqref{KS}, but due to aggressive filtering,
{cf.~\eqref{eq:hG},} is unnecessary in the latter two problems.
Evaluation of the model terms in \eqref{eq:KSLES} and
\eqref{eq:aKSLES} requires differentiation of the eddy viscosity
$\nu(|{s}|)$ with respect to ${s}$ which is performed using
spectrally-accurate Chebyshev differentiation matrices
\cite{trefethen2000spectral} defined in the state domain $\LL$. The
eddy viscosity $\nu(|{s}|)$ and its derivative $d \nu(|{s}|) / d{s}$
are then interpolated from the state space to the physical space using
the barycentric formulas which are also spectrally accurate
\cite{trefethen2013approximation}. {This step ensures that the
  regularity required of the eddy viscosity, cf.~\eqref{eq:nue_reg},
  is maintained.} The time-discretization of systems \eqref{KS},
\eqref{eq:KSLES}, and \eqref{eq:aKSLES} is performed using the
exponential time-differencing fourth-order Runge-Kutta method (ETDRK4)
\cite{kassam2005fourth}, originally introduced in
\cite{cox2002exponential}, which is fourth-order accurate. The
different integrals are approximated using Gaussian quadratures, which
are given by the trapezoidal rule on the periodic domain $\Omega$
(e.g., in \eqref{eq:J}), and by the Clenshaw-Curtis formula on the
bounded domain $\LL$ (e.g., in \eqref{eq:ip}). The boundary-value
problem {\eqref{eq:gradH3bc}--}\eqref{eq:gradH3} defining the Sobolev
gradients is solved using the {\tt chebop} feature of Chebfun
\cite{driscoll2014chebfun}, where the discretization is performed
based on ultraspherical polynomials. With most computations carried
out with spectral accuracy, approximation errors are dominated by
time-stepping errors where the accuracy is $\mathcal{O}((\Delta
t)^4)$. Unless mentioned otherwise, in our computations we use $N_x =
{1024}$, $N_s = {4096}$ and $\Delta t = {3.0 \times 10^{-6}}$.

In order to validate the discretization techniques discussed
  above, we verify the accuracy of the cost functional gradients
  evaluated as described in Section \ref{sec:adjoint},
  cf.~\eqref{eq:gradL2}, by computing the G\^ateaux differential
  $\J'(\nu;\nu')$ in terms of the Riesz identity \eqref{eq:Riesz} and
  comparing it with its approximation obtained with a simple forward
  finite-difference formula. The ratio of these two expressions is
  thus given by
\begin{equation}
\kappa(\epsilon) := \frac{\epsilon^{-1} \left[ \J(\nu + \epsilon \nu') - \J(\nu) \right]}
{\left\langle \grad_{\nue}^{L^2}\mathcal{J}, \nu' \right\rangle_{L^2(\LL)}},
\label{eq:kappa}
\end{equation}
where $\nu'$ is an arbitrary perturbation and $\epsilon > 0$ its
magnitude. We expect $\kappa(\epsilon)$ to be close to unity and this
is indeed evident in Figure \ref{fig:kappa} for a range of $\epsilon$
spanning several orders of magnitude. The large deviations of
$\kappa(\epsilon)$ from unity observed for very small and very large
values of $\epsilon$ are due to, respectively, the round-off and
truncation errors in the finite-difference formula, both of which are
well-known effects \cite{bukshtynov2013optimal,bukshtynov2011optimal}.
In Figure \ref{fig:Kap_log} we also note that, as expected, for intermediate
values of $\epsilon$, $\kappa(\epsilon) \rightarrow 1$ as the
discretization parameter $\Delta t$ used in the numerical solution of
the PDE systems \eqref{eq:KSLES} and \eqref{eq:aKSLES} is refined.
These results demonstrate the {consistency} of the cost
functional gradients evaluated as discussed in Section
\ref{sec:adjoint} {and also show that when sufficient numerical
  resolution is used, discretization errors will have a vanishing
  effect on the accuracy of the gradients and therefore also on the
  accuracy of the obtained optimal forms of the eddy viscosity. Thus,
  with the values of the discretization parameters $N_x$, $N_s$ and
  $\Delta t$ indicated above, our computations are fully resolved such
  that further refinements of these parameters would not produce
  appreciable changes of the results.}
\begin{figure}
  \centering
\mbox{
  \subfigure[]
  {
    \includegraphics[scale=0.44]{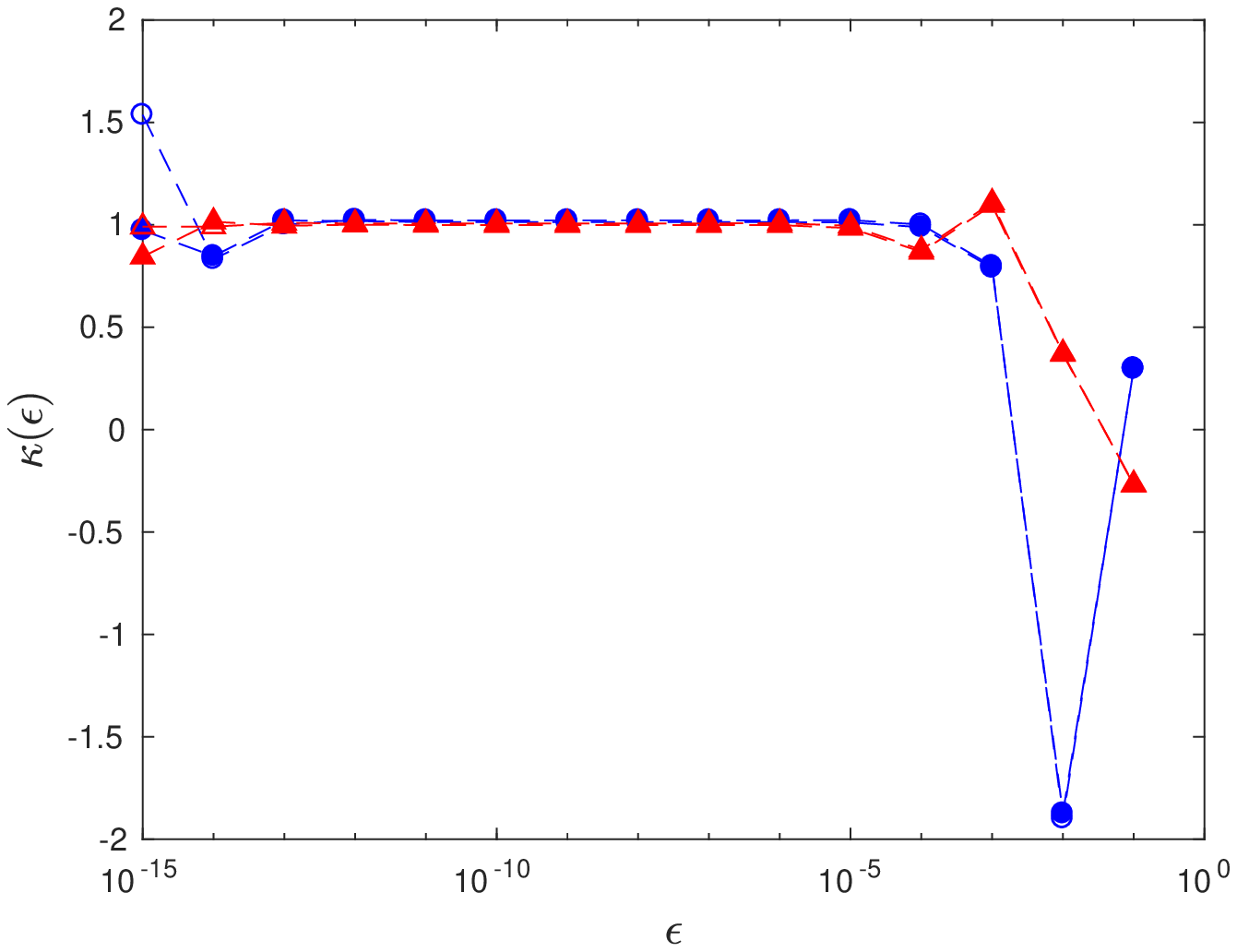}
    \label{fig:Kap}
  }\qquad\qquad
  \subfigure[]
  {
    \includegraphics[scale=0.44]{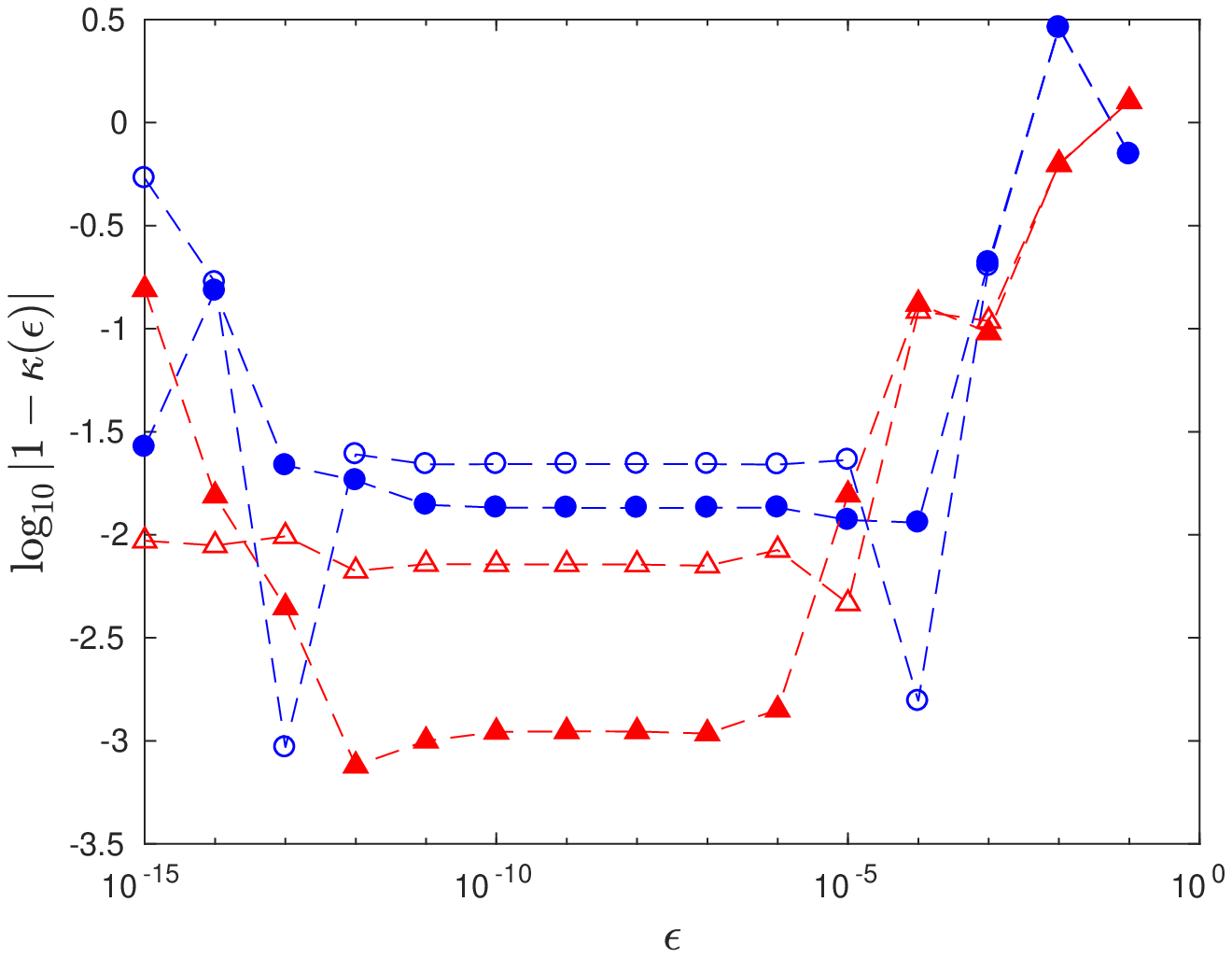}
    \label{fig:Kap_log}
  }
}
\caption{Dependence of \subref{fig:Kap} $\kappa$ and
    \subref{fig:Kap_log} $\log_{10}|1 - \kappa|$ on $\epsilon$,
    cf.~\eqref{eq:kappa}, for two different perturbations $\nu'$ (blue
    circles vs.~red triangles) and two different time steps (empty
    symbols) $\Delta t = 3.0 \times 10^{-6}$ and (filled symbols)
    $\Delta t = 1.0 \times 10^{-6}$ used in the solution of the PDE
    problems \eqref{eq:KSLES} and \eqref{eq:aKSLES}.}
\label{fig:kappa}
\end{figure}

\subsection{Gradient Descent}
\label{sec:desc}

While for simplicity of presentation in Section \ref{sec:optimal} the
steepest-descent (simple gradient) approach was used in
\eqref{eq:desc}, in actual computations we use the conjugate-gradients
method \cite{nw00} which is known to significantly accelerate
convergence. The gradient $\grad_{\nue}\J(\nu^{(n)})$ in
\eqref{eq:desc} is then replaced with the descent direction $g^{(n)}$
defined as
\begin{equation}
\begin{aligned}
g^{(n)} &= \grad_{\nue}\J(\nu^{(n)}) - \gamma_n \, g^{(n-1)}, \quad n=1,2,\dots, \\
g^{(0)} &= \grad_{\nue}\J(\nu^{(0)}),
\end{aligned}
\label{eq:CG}
\end{equation}
where the ``momentum'' term $\gamma_n$ is evaluated using the
Polak-Ribi{\`{e}}re formula
\begin{equation} 
\label{eq:CGPR}
\gamma_n = \frac{\Big\langle \left(\grad_{\nue}\J(\nu^{(n)})  - \grad_{\nue}\J(\nu^{(n-1)}) \right), \grad_{\nue}\J(\nu^{(n)}) \Big\rangle_{H^3(\LL)}}
{\Big\langle \grad_{\nue}\J(\nu^{(n-1)}), \grad_{\nue}\J(\nu^{(n-1)}) \Big\rangle_{H^3(\LL)} }.
\end{equation}
It is a good practice \cite{nw00} for the conjugate-gradients approach
\eqref{eq:CG} to be periodically restarted {with a gradient step}
after a certain number of iterations. The line-minimization problem
\eqref{eq:taun} is {efficiently} solved using Brent's algorithm
\cite{pftv86}, which is a standard approach. Gradient iterations
\eqref{eq:desc} are declared converged when the following termination
condition is satisfied
\begin{equation} 
\label{eq:iterend}
\frac{|\J(\nu^{(n+1)}) - \J(\nu^{(n)})|}{\J(\nu^{(n)})} < \epsilon_{\J},
\end{equation}
where $\epsilon_{\J}$ is a prescribed tolerance (we will use
$\epsilon_{\J} = 10^{-7}$). The choice of the length-scale parameters
$\ell_1$, $\ell_2$ and $\ell_3$ defining the Sobolev inner product
\eqref{eq:ipH3} will be discussed in the next section. The different
steps in the solution of the optimization problem \eqref{eq:minJ} are
summarized as Algorithm \ref{algo1}.

\begin{algorithm}[t] 
  \caption{Implementation of the conjugate-gradients variant of
    descent approach \eqref{eq:desc}.
  \newline \textbf{Input:} \\
  \hspace*{0.3cm} $\{m(t)\}_{i=1}^N$ --- target observations (e.g., of the DNS, cf.~\eqref{KS}; $N$ is number of observations)  \\
  \hspace*{0.3cm} $N_x,N_s,\Delta t$ --- numerical discretization parameters \\
  \hspace*{0.3cm} $\ell_1,\ell_2,\ell_3$ --- Sobolev length scales, cf.~\eqref{eq:ipH3} \\
  \hspace*{0.3cm} $\epsilon_{\J}$ --- tolerance in the termination criterion \eqref{eq:iterend} \\
  \hspace*{0.3cm} $\nu_0$ --- initial guess for the eddy viscosity (e.g., \eqref{eq:nusKS}) \\
  \textbf{Output:} \\
  \hspace*{0.3cm} $\nud$ --- optimal eddy viscosity \\ 
}
\label{algo1}
\begin{algorithmic}
\State \textbullet~set $n = 0$
\State \textbullet~set $\nu^{(0)} = \nu_0$
\Repeat 
\State \textbullet~set $n = n+1$
\State \textbullet~solve the LES problem \eqref{eq:KSLES}
\State \textbullet~solve the adjoint problem \eqref{eq:aKSLES}
\State \textbullet~determine the {$L^2$} gradient $\grad_{\nu}^{L^2}\J(\nu^{(n)})$, cf.~\eqref{eq:gradL2}
\State \textbullet~determine the Sobolev gradient $\grad_{\nu}^{H^3}\J(\nu^{(n)})$ via {\eqref{eq:gradH3bc}--}\eqref{eq:gradH3}
\State \textbullet~determine the ``momentum'' coefficient $\gamma_n$ via \eqref{eq:CGPR}
\State \textbullet~determine the conjugate descent direction $g^{(n)}$ via \eqref{eq:CG}
\State \textbullet~determine the optimal step length $\tau^{(n)}$ by solving \eqref{eq:taun} {with Brent's algorithm}
\State \textbullet~update the eddy viscosity $\nu^{(n)} \ \leftarrow \nu^{(n)} - \tau_n \, g^{(n)}$
\Until termination criterion \eqref{eq:iterend} is satisfied
\end{algorithmic}
\end{algorithm}

\section{\label{sec:Results}Results} 

In this section we present and analyze the results obtained with our
approach to determining optimal {eddy-viscosity} closure models,
cf.~Algorithm \ref{algo1}, and compare them to the results obtained
with other closure models including the approach of Das \& Moser
\cite{das2002optimal}, where the closure model also has some
optimality properties.  We consider observation operators $H_i$,
$i=1,\dots,N$, defined both in the physical space, cf.~Section
\ref{sec:Hph}, and in the Fourier space, cf.~Section \ref{sec:Hs}, the
latter with different distributions of the observed wavenumbers
\eqref{eq:KK}.  With most other problem parameters fixed as discussed
in Section \ref{sec:params}, there remains one key characteristic
defining the optimization problem \eqref{eq:J}--\eqref{eq:minJ},
namely, the length of the time window and different values of $T$ are
considered, covering from a few to several typical events in the
evolution of the Kuramoto-Sivashinsky system \eqref{KS} (these
``events'' are the merging of crests or formation on new ones).
Summary information about these computations is compiled in Table
\ref{tab:summary}, where we also indicate the values of the
length-scale parameters $\ell_2$ and $\ell_3$ ($\ell_1 = 0$) used to
determine the Sobolev gradients, cf.~\eqref{eq:ipH3}.  The values of
these parameters are chosen by trial and error to maximize the
convergence of iterations in Algorithm \ref{algo1}.  We conclude from
the data compiled in Table \ref{tab:summary} that in all cases
optimization reduces the observation error \eqref{eq:J} by a factor
$\mathcal{O}(1)$ with its specific value depending on the length of
the optimization window $T$ and there tends to be an optimal value of
$T$ for which the largest reduction of the error functional
\eqref{eq:J} is achieved. To fix attention, we will henceforth focus
on two representative configurations, namely, one with observations in
the physical space and one with observations in the Fourier space,
denoted respectively ``Case A'' and ``Case B'' in Table
\ref{tab:summary}.

\begin{table}[t!]
\centering
\resizebox{\textwidth}{!}{\begin{tabular}{ |c|c|c|c|c|c|c|c|c|c| } 
 \hline
{Observations}
 & \multicolumn{3}{|c|}{Physical Space} & \multicolumn{3}{|c|}{Fourier Space --- Equispaced} & \multicolumn{3}{|c|}{Fourier Space --- Clustered} \\
 \hline
 \text{T} & $1.5 \times 10^{-3}$ & $3.0 \times 10^{-3}$ & $9.0 \times 10^{-3}$ & $1.5 \times 10^{-3}$ & $3.0 \times 10^{-3}$ & $9.0 \times 10^{-3}$ & $1.5 \times 10^{-3}$ & $3.0 \times 10^{-3}$ & $9.0 \times 10^{-3}$ \\ 
 \hline
 ${\ell_{2}}$ & {$10^3$} & {$10^3$} & {$10^6$} & {$10^3$} & {$10^3$} & {$10^5$} & {$10^3$} & {$10^3$} & {$10^3$} \\ 
 \hline
 ${\ell_{3}}$ & {$10^1$} & {$10^1$} & {$10^5$} & {$10^1$} & {$10^1$} & {$10^5$} & {$10^1$} & {$10^1$} & {$10^1$} \\
 \hline
 ${ \frac{\J(\nu_0)}{\J(\nud)} }$ & {$1.51$} & {$8.21$} & {$1.80$} & {$3.51$} & {$2.75$} & {$1.49$} & {$5.02$} & {$6.76$} & {$2.22$} \\
 \hline
 Note &  & Case A &  &  &  &  &  & Case B &  \\
 \hline
\end{tabular}}
\caption{Summary information about the different cases considered in our computations.}
\label{tab:summary}
\end{table} 

The decrease of the normalized objective functional \eqref{eq:J} with
iterations $n$ in Cases A and B is shown in Figure \ref{fig:J}. In
this figure we observe a reduction of the observation error by close
to one order of magnitude over $\mathcal{O}(10)$ iterations. We note
however that the convergence rate of iterations \eqref{eq:desc} is
rather nonuniform. The corresponding optimal eddy viscosities
$\nud(s)$ are shown as functions of the resolved strain $s$ together
with the Smagorinsky model \eqref{eq:nusKS} used as the initial guess
in \eqref{eq:desc} in Figure \ref{fig:nud}. We see that while the
optimal eddy viscosities $\nud(s)$ are defined on a larger domain
$\LL$, the deviations from the initial guess $\nu_0$ produced by the
gradient iterations \eqref{eq:desc} are essentially confined to a
smaller identifiability interval $\I$. Most importantly, in contrast
to the original Smagorinsky model \eqref{eq:nusKS}, the optimal eddy
viscosities are negative for small strains such that $\nud(0) < 0$. It
is encouraging to note that the optimal eddy viscosities obtained in
Cases A and B exhibit a qualitatively similar dependence on $s$,
despite rather different forms of observations used to define the
optimization problem \eqref{eq:J}--\eqref{eq:minJ} in these two cases.
\begin{figure}
  \centering
\mbox{
  \subfigure[]
  {
    \includegraphics[scale=0.44]{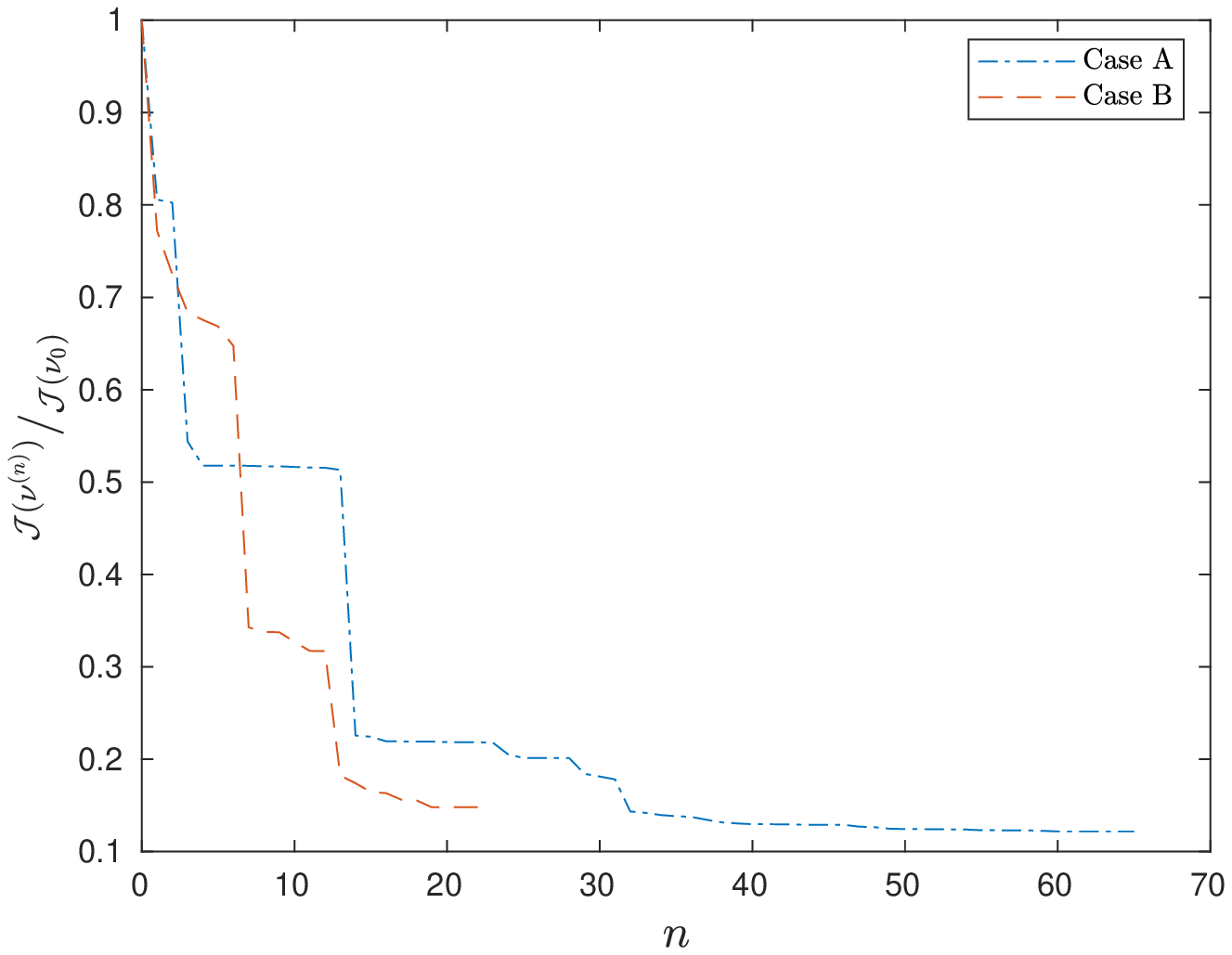}
    \label{fig:J}
  } \qquad\qquad
  \subfigure[]
  {
    \includegraphics[scale=0.44]{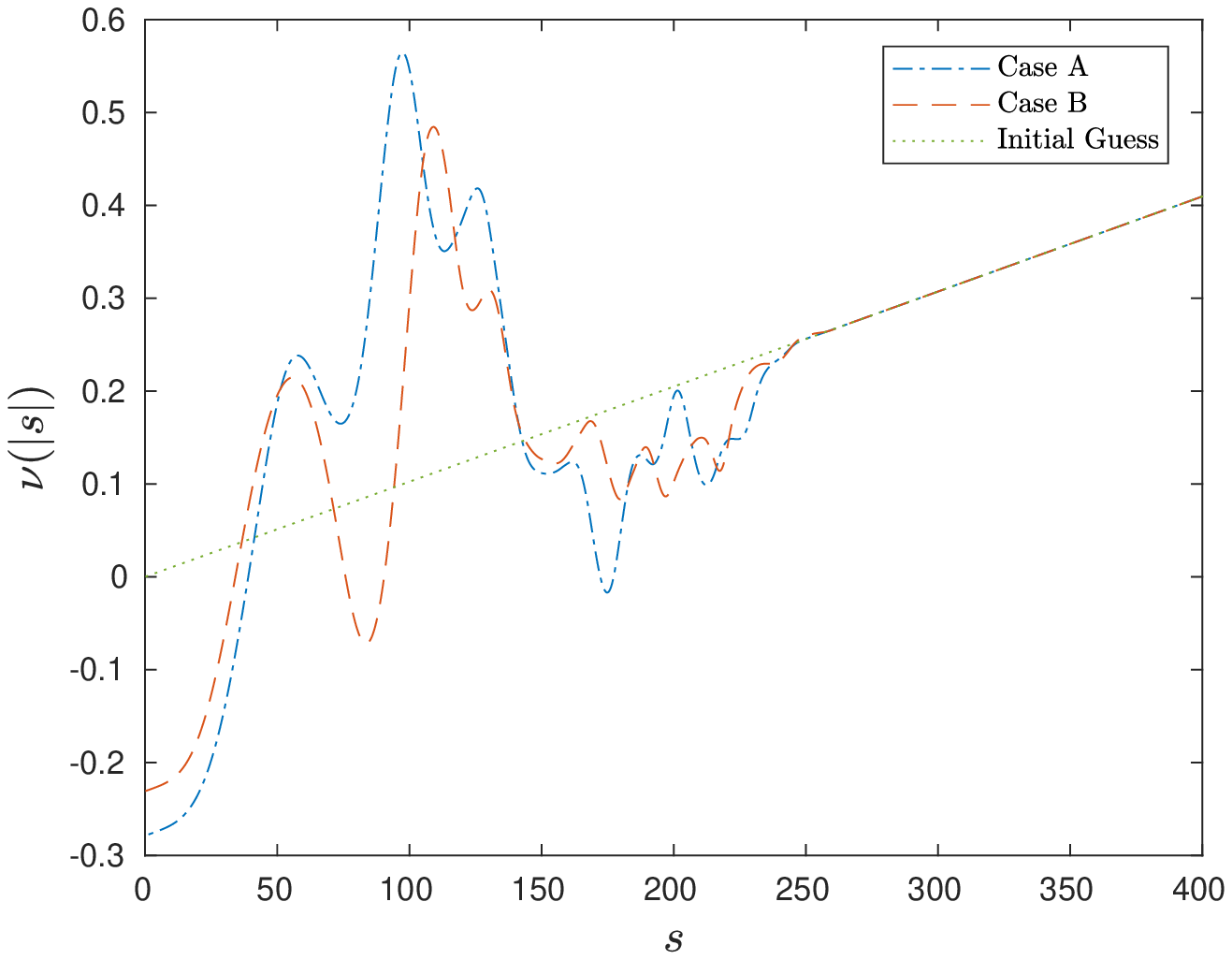}
    \label{fig:nud}
  }}
\caption{(a) Dependence of the normalized error functional
  $\J(\nu^{(n)}) / \J(\nu_0)$ on the iteration count $n$ and (b)
  dependence of the optimal eddy viscosity $\nud({|s|})$ on the
  resolved strain $s$ for Case A (blue dash-dotted line) and Case B
  (red dashed line), cf.~Table \ref{tab:summary}. In panel (b) we also
  indicate the Smagorinsky model \eqref{eq:nusKS} used as the initial
  guess $\nu_0$ (green dotted line).}
\end{figure}

In order to obtain insights about the spatio-temporal evolution of
solutions to the LES problems \eqref{eq:KSLES} with different closure
models (no closure at all, the Smagorinsky model \eqref{eq:nusKS}, and
the optimal eddy viscosity $\nud({|s|})$ from Cases A and B,
cf.~Table \ref{tab:summary}), in Figure \ref{fig:utx} these evolutions
are compared as functions of space and time to the DNS solution
$w(t,x)$ of the original Kuramoto-Sivashinsky system \eqref{KS}. In
this figure we also include the evolution obtained with the optimal
closure model proposed Das \& Moser \cite{das2002optimal} based on a
stochastic estimator. In order to assess the performance of the
proposed approach at times $t > T$ extending beyond the ``training
window'' $[0,T]$, the evolutions in Figure \ref{fig:utx} are shown for
$t \in [0,2T]$. We observe that with the exception of the LES solution
with no closure model, cf.~Figure \ref{fig:utx_NoClosure}, all LES
solutions are qualitatively quite similar to the DNS solution,
especially for short times, cf.~Figures
\ref{fig:utx_IGuess}--\ref{fig:utx_B} vs.~Figure \ref{fig:utx_DNS},
although the solutions obtained for Cases A and B arguably best
correlate with the DNS results.
\begin{figure}\centering
  \subfigure[]
  {
    \includegraphics[scale=0.44]{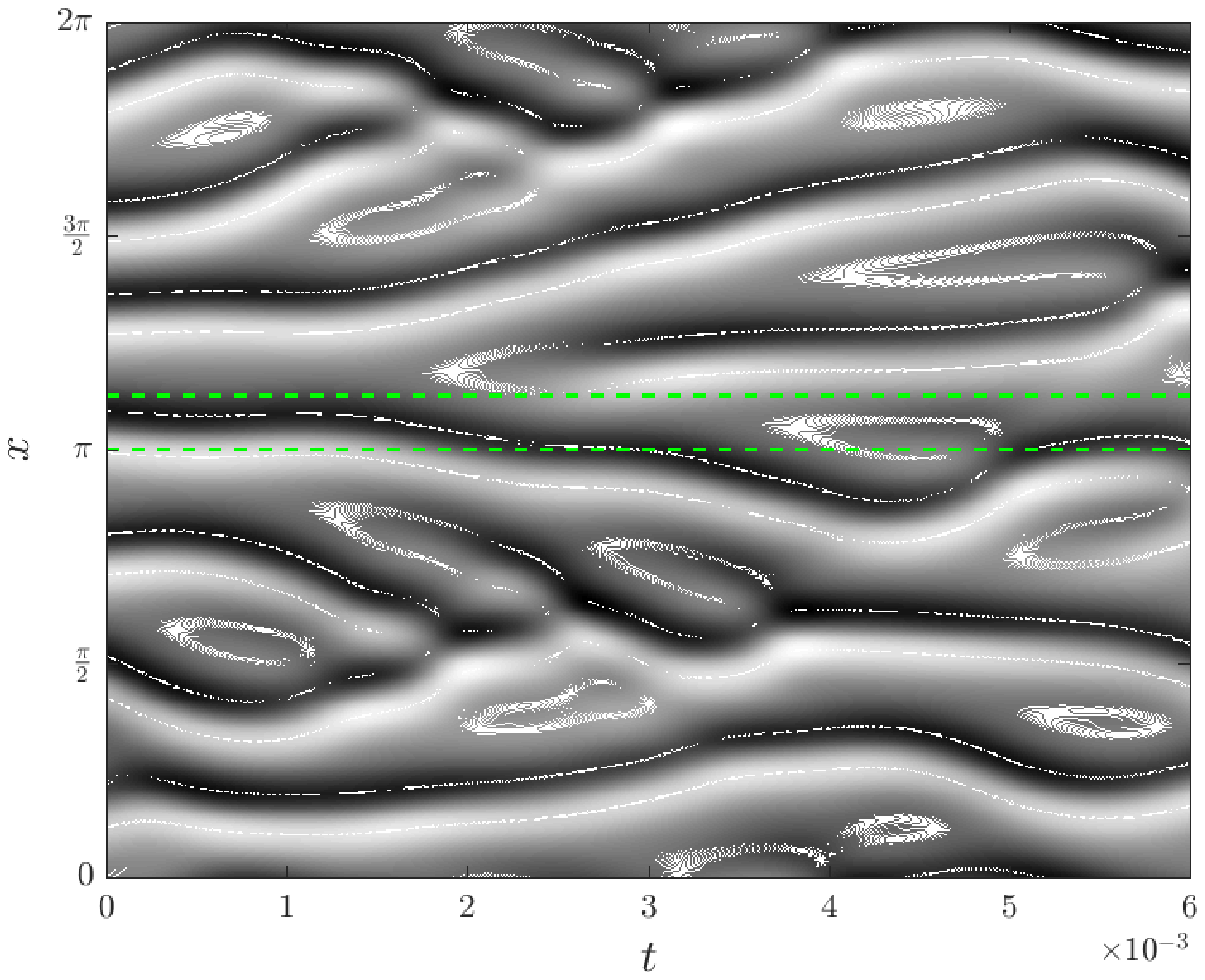}
    \label{fig:utx_DNS}
  }\qquad
  \subfigure[]
  {
    \includegraphics[scale=0.44]{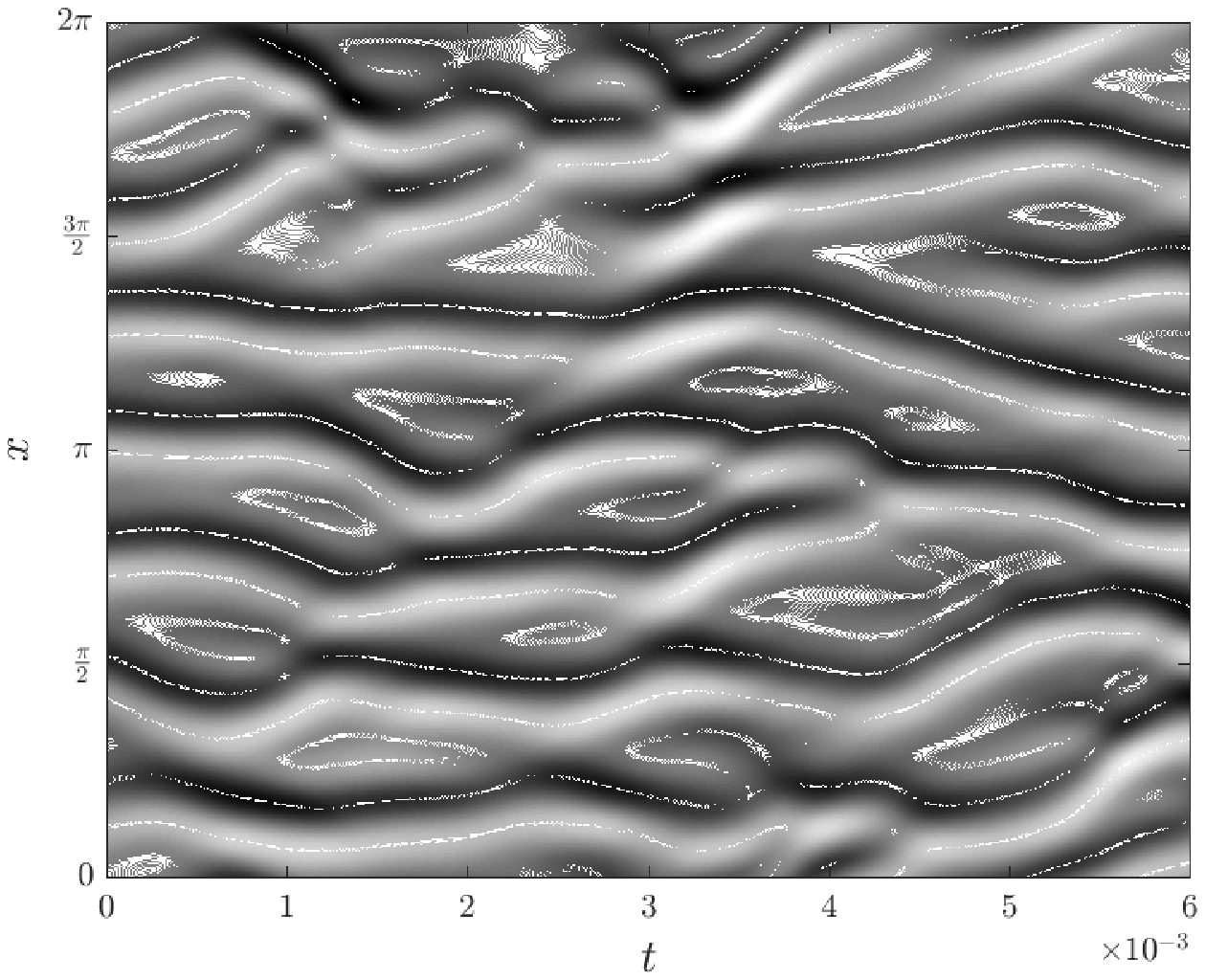}
    \label{fig:utx_NoClosure}
  }
  \subfigure[]
  {
    \includegraphics[scale=0.44]{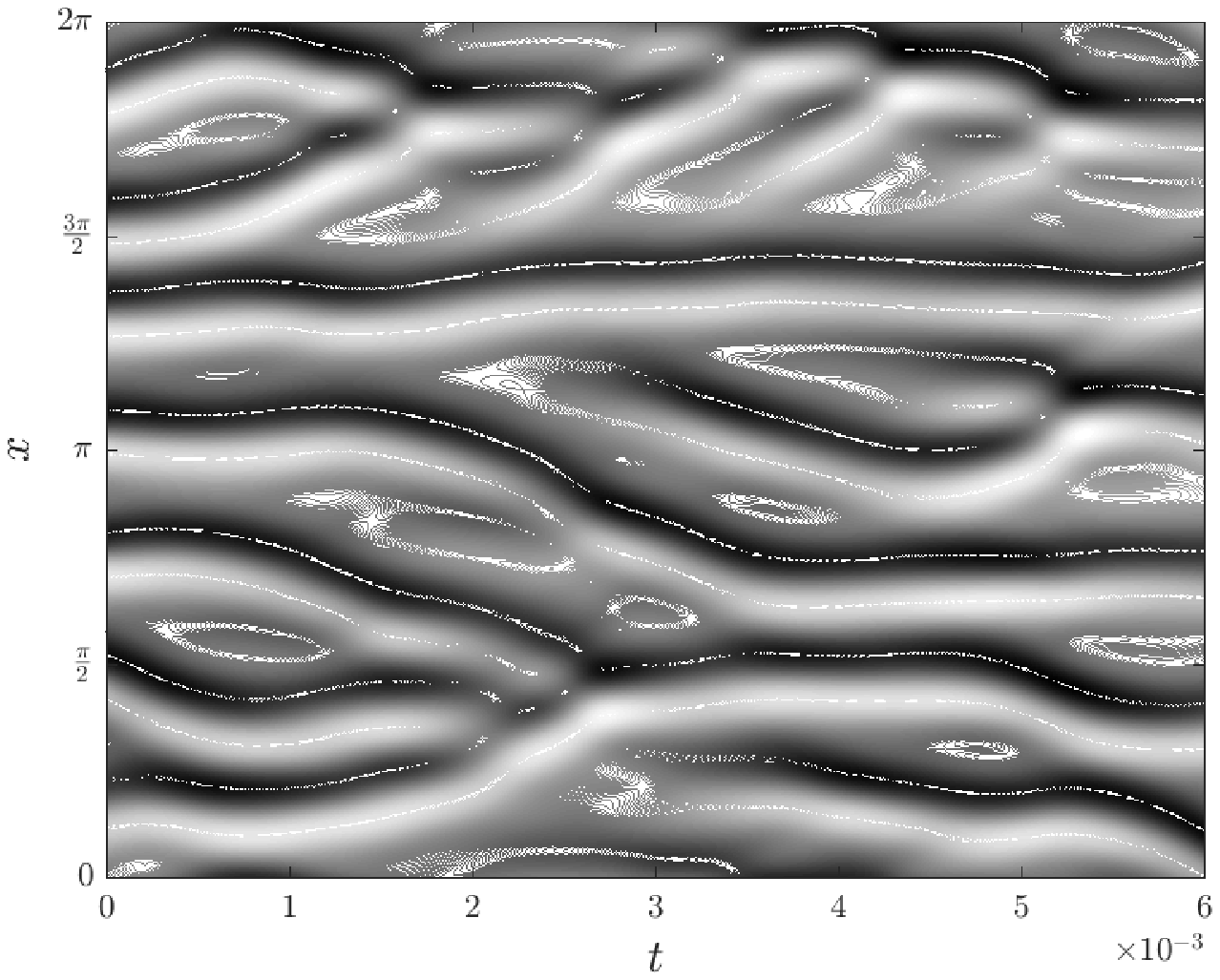}
    \label{fig:utx_IGuess}
  }\qquad
  \subfigure[]
  {
    \includegraphics[scale=0.44]{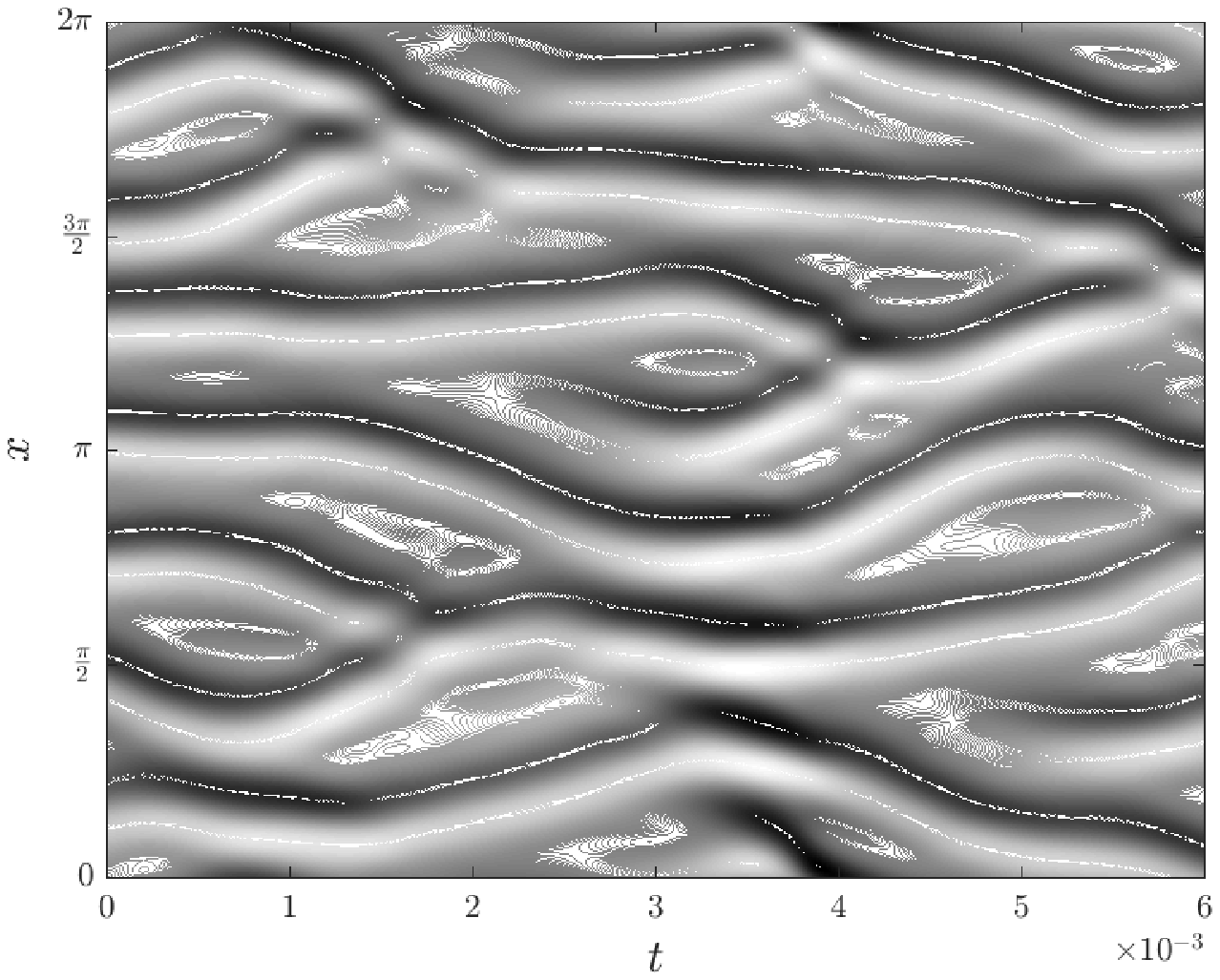}
    \label{fig:utx_Moser}
  }
  \subfigure[]
  {
    \includegraphics[scale=0.44]{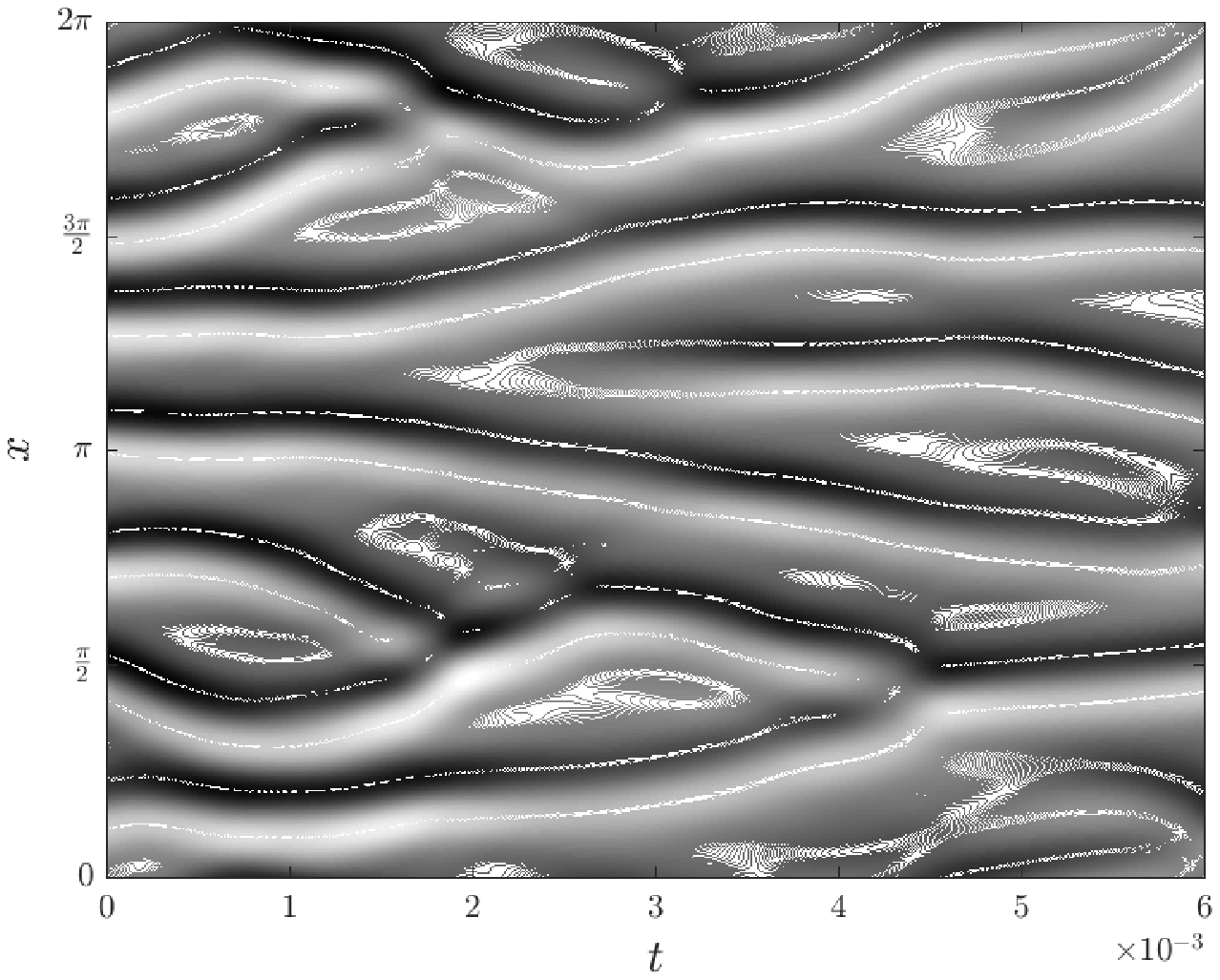}
    \label{fig:utx_A}
  }\qquad
  \subfigure[]
  {
    \includegraphics[scale=0.44]{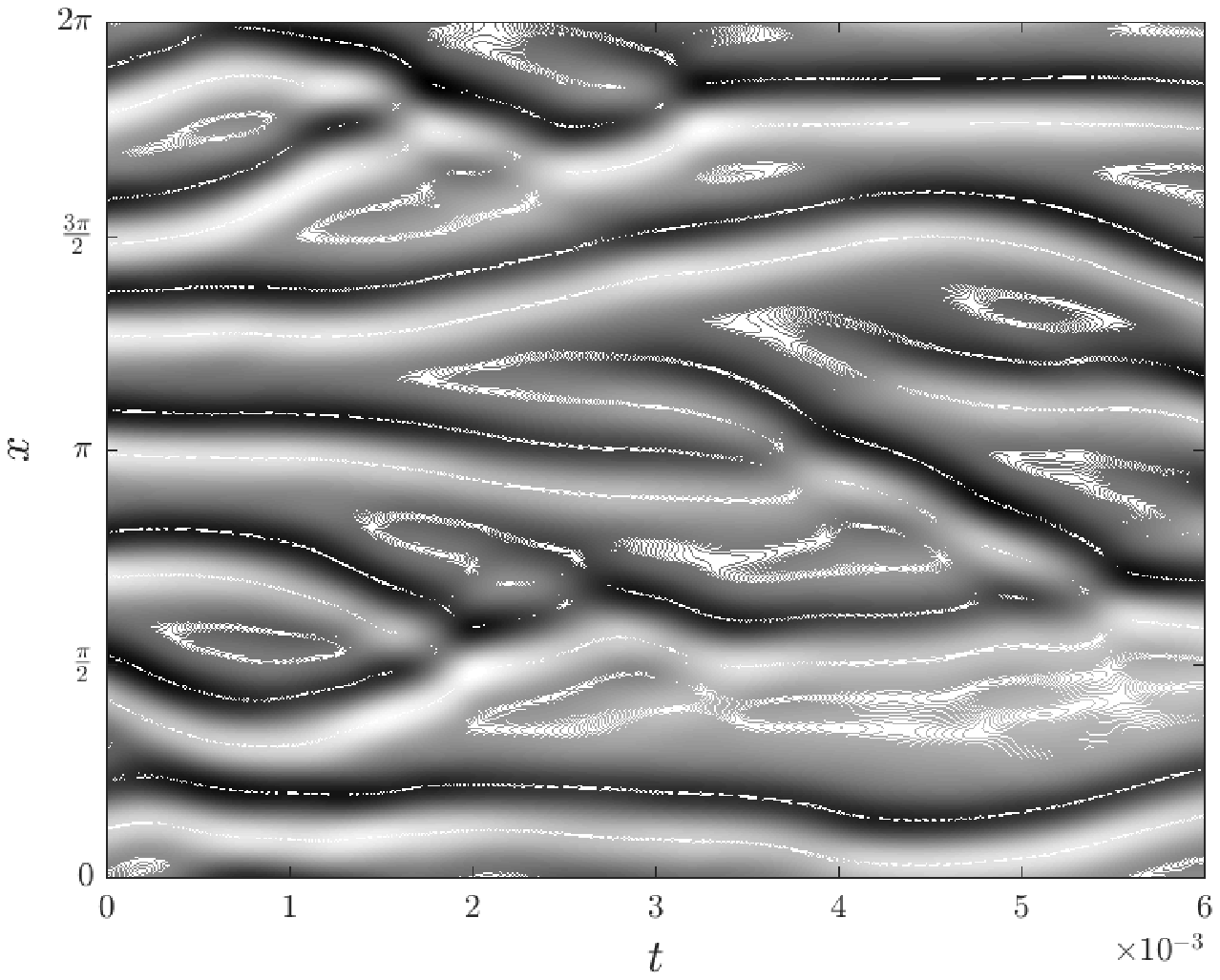}
    \label{fig:utx_B}
  }
  \caption{Space-time evolution of solutions to \subref{fig:utx_DNS}
    the DNS problem \eqref{KS} and the LES problems \eqref{eq:KSLES}
    with \subref{fig:utx_NoClosure} no closure model, with
    \subref{fig:utx_IGuess} the Smagorinsky model \eqref{eq:nusKS},
    with \subref{fig:utx_Moser} the optimal closure proposed by Das \&
    Moser \cite{das2002optimal}, with \subref{fig:utx_A} the optimal
    closure from case A, and with \subref{fig:utx_B} the optimal
    closure from case B. Grayscale indicates the solution value at
    given $(x,t)$. In panel \subref{fig:utx_DNS} the green, dashed
    horizontal lines indicate the cut-off length scale $\delta$
    characterizing filter \eqref{eq:hG}.}
\label{fig:utx}
\end{figure}

\begin{figure}\centering
  \subfigure[]
  {
    \includegraphics[scale=0.44]{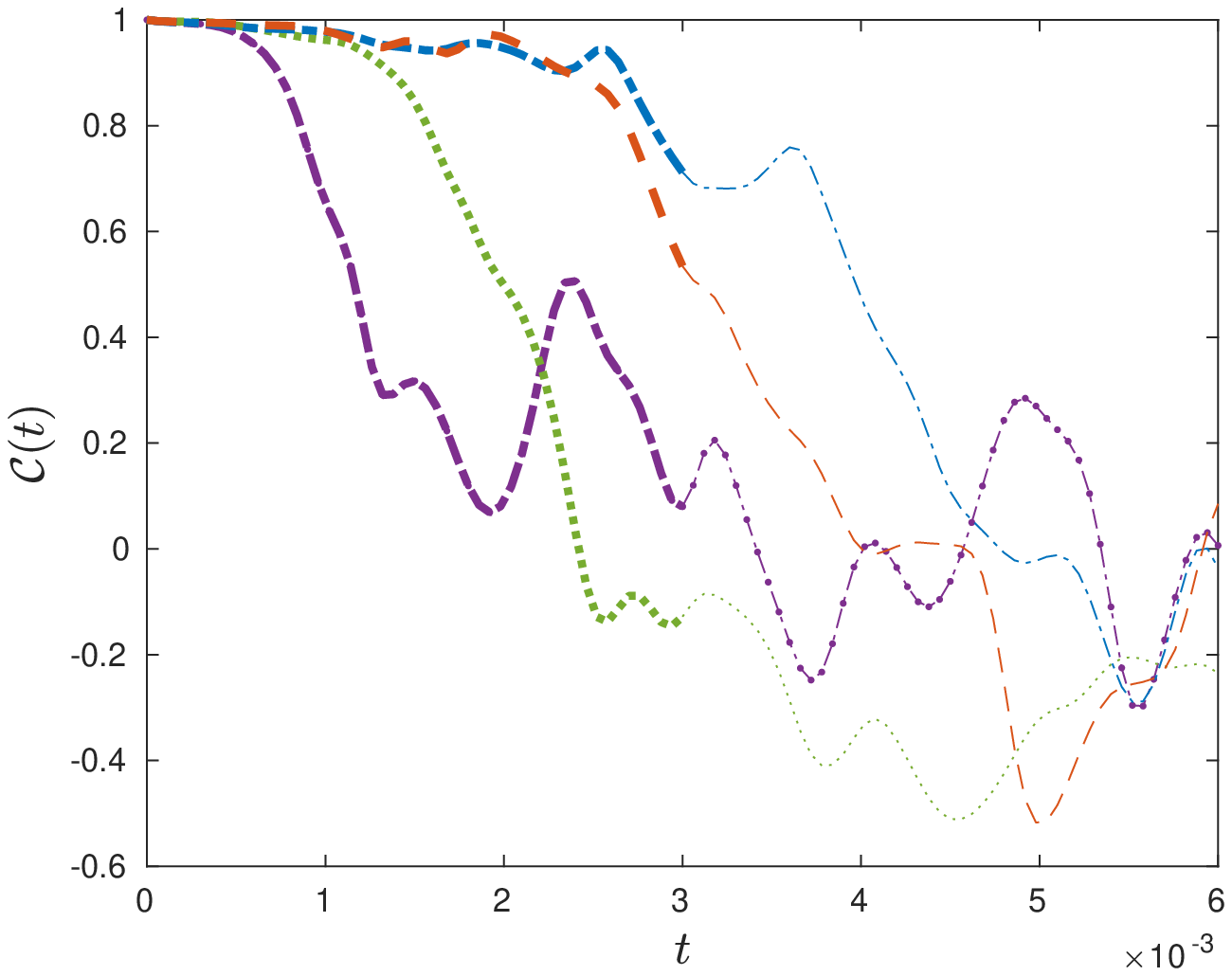}
    \label{fig:Err1_1}
  }\qquad\qquad
  \subfigure[]
  {
    \includegraphics[scale=0.44]{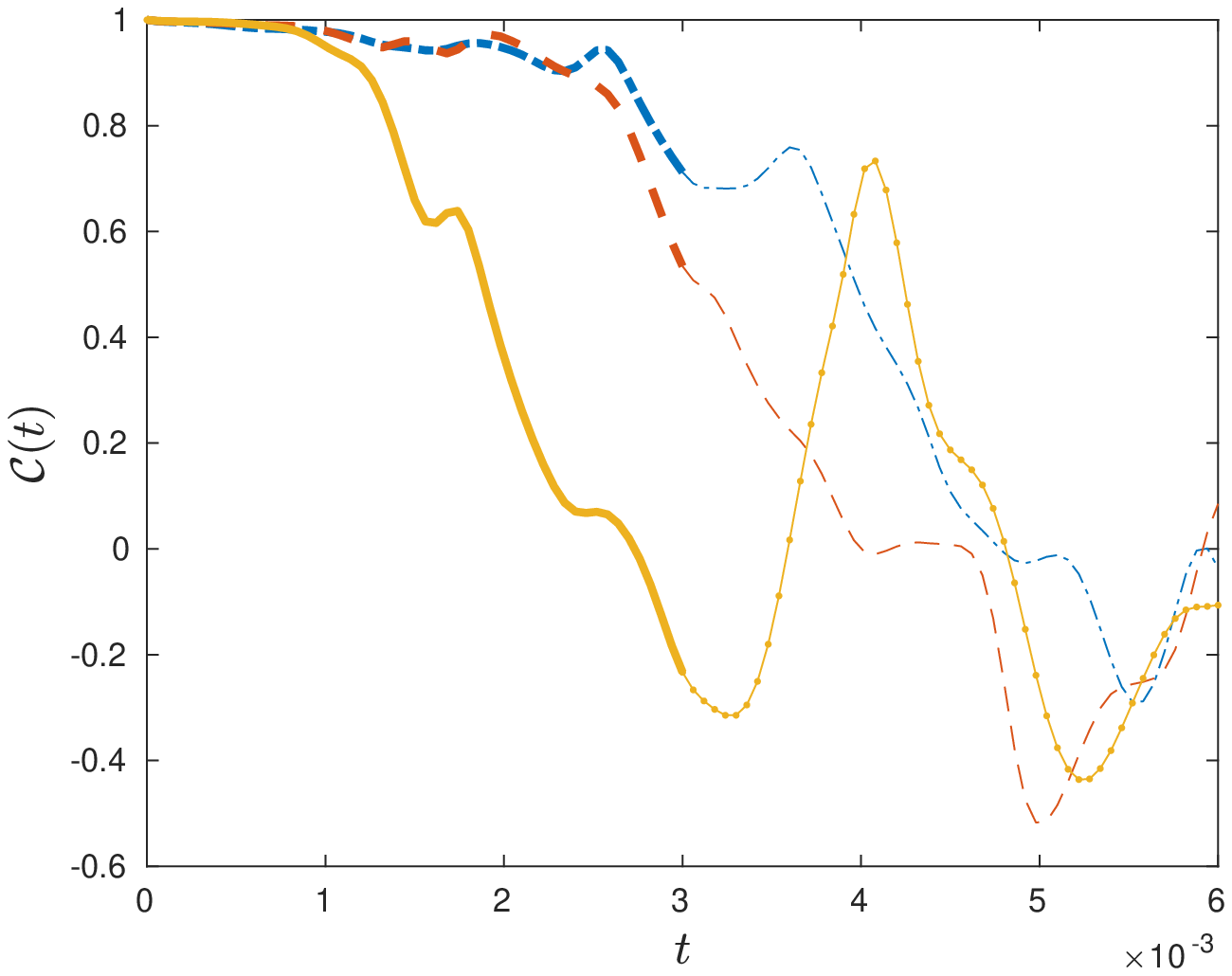}
    \label{fig:Err1_2}
  }
  \subfigure[]
  {
    \includegraphics[scale=0.44]{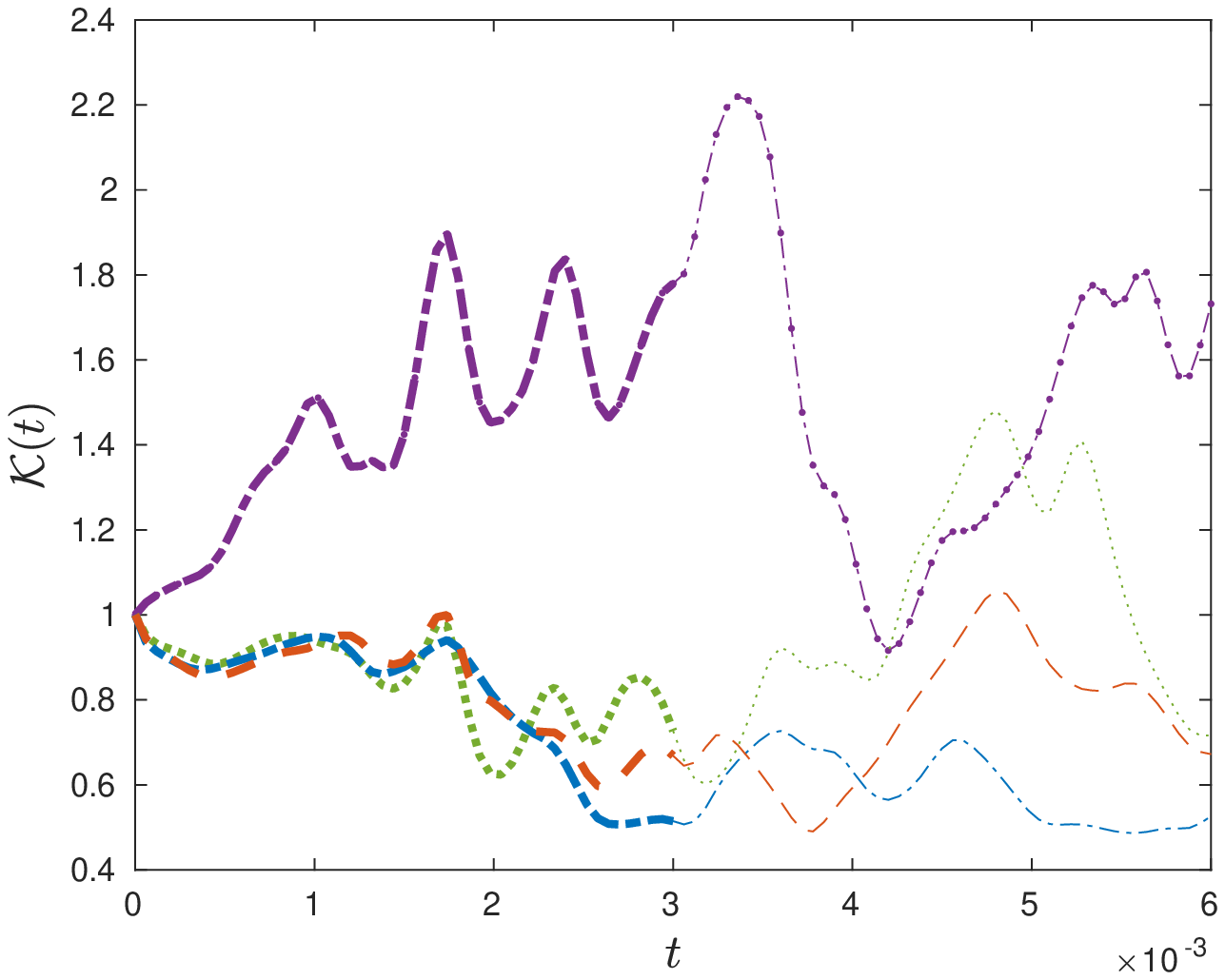}
    \label{fig:Kin_1}
  }\qquad\qquad
  \subfigure[]
  {
    \includegraphics[scale=0.44]{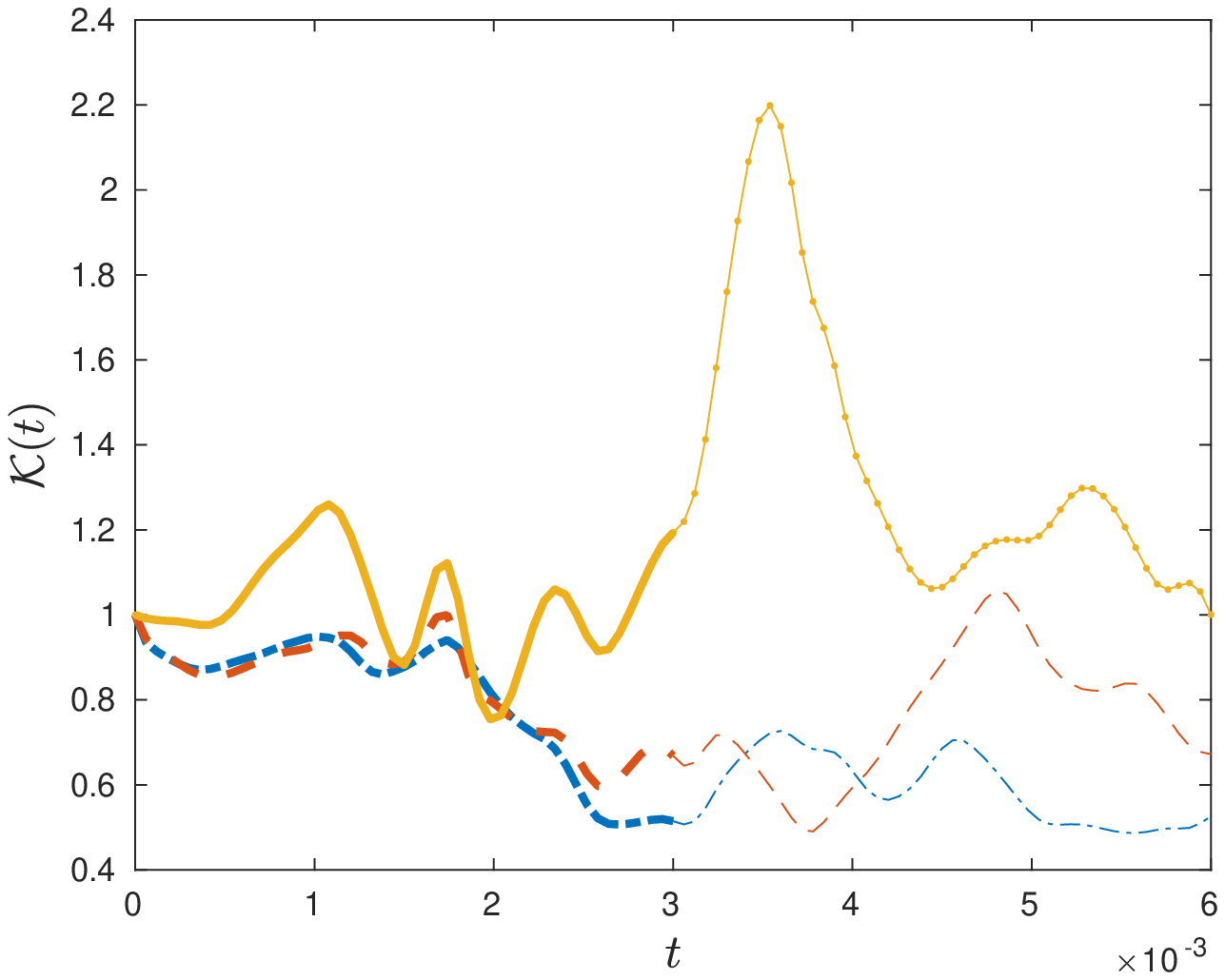}
    \label{fig:Kin_2}
  }
  \caption{Dependence of the diagnostic quantities (a,b) $\C(t)$,
    cf.~\eqref{eq:C}, and (c,d) $\K(t)$,
    cf.~\eqref{eq:K}, 
    on time $t$ for the LES solutions with (purple dot-dashed line) no
    closure model, with (green dotted line) the Smagorinsky model
    \eqref{eq:nusKS}, with (blue dash-dotted line) the optimal eddy
    viscosity $\nud$ from case A, with (red dashed line) the optimal
    eddy viscosity $\nud$ from case B, and with (yellow solid-dotted
    line) the optimal closure proposed by Das \& Moser
    \cite{das2002optimal}. {Bold and thin lines correspond to time in
      the ``training window'' $t \in [0, T]$ and beyond the ``training
      window'' $t \in (T, 2T]$, respectively.}}
  \label{fig:diagn}
\end{figure}

To quantify these observations, we will now analyze the behavior of
the following two diagnostic quantities
\begin{subequations} 
\label{eq:diagn}
\begin{align}
\C(t) & := \frac{1}{||w(t)||_{L^2(0, 2\pi)} ||\tu(t)||_{L^2(0, 2\pi)} } \int_{0}^{2\pi} w(t, x) \, \tu(t, x) \ dx, \label{eq:C} \\
\K(t) & := \frac{||\tu(t)||^2_{L^2(0, 2\pi)}}{||w(t)||^2_{L^2(0, 2\pi)}} = \frac{1}{||w(t)||^2_{L^2(0, 2\pi)}} \, \int^{2\pi}_{0} \tu(t, x)^2 \, dx, \label{eq:K} 
\end{align}
\end{subequations}
which can be interpreted as, respectively, the correlation of the LES
solution $\tu$ with the DNS solution $w$ and the normalized kinetic
energy. 
The {LES with the} optimal {eddy-viscosity} closures obtained in Cases
A and B are compared in terms these diagnostic quantities for $t \in
[0,2T]$ to the LES with no closure model, with the Smagorinsky model
\eqref{eq:nusKS}, and with the closure model of Das \& Moser
\cite{das2002optimal} in Figure \ref{fig:diagn}.  Given the chaotic
nature of the Kuramoto-Sivashinsky system \eqref{KS} resulting in an
exponentially fast divergence of initially nearby trajectories, in all
cases the correlation \eqref{eq:C} drops very rapidly, such that for
short times $t \rightarrow 0$ we approximately have $\C(t) \approx 1 -
r \, e ^{pt}$ for some $r,p > 0$ different in each case, cf.~Figures
\ref{fig:diagn}(a).  The growth rate $p$, characterizing the
exponential divergence of the solutions to the DNS and LES problems,
is smallest when the optimal {eddy-viscosity} closure {models
  $\nud(|s|)$ from cases A and B are} used. As a result, the time
$t_0$ when the DNS and the LES solutions become uncorrelated, i.e.,
when $\C(t_0) \approx 0$, is nearly twice as large for the optimal
{eddy-viscosity} closures from Cases A and B than for the Smagorinsky
model \eqref{eq:nusKS} and the closure model of Das \& Moser
\cite{das2002optimal}.  As regards the behavior of the normalized
energy \eqref{eq:K}, 
from Figures \ref{fig:Kin_1}--\ref{fig:Kin_2} we see that the optimal
eddy viscosity $\nud({|s|})$ on average tends to reduce the kinetic
energy {relative to its levels in the DNS.}  This is in contrast to
the approach of Das \& Moser \cite{das2002optimal} in which the
normalized energy is increased to levels higher than in the DNS,
cf.~Figure \ref{fig:Kin_2}.

While the analysis above focused on ``a posteriori'' tests involving
the results of solving the LES system \eqref{eq:KSLES} with different
closure models, we close this section with a brief discussion of an
``a priori'' test where the errors in approximations of the SGS
stresses \eqref{eq:Mw} with different closure models are analyzed.
{In this context,} comparing the SGS dissipation rate with the
modeled SGS dissipation rate is often used as a {standard
  diagnostic} to assess the energetics in the LES flow
\cite{Meneveau1994}. We thus focus on the following normalized
least-squares measure of this error
\begin{equation} 
\mathcal{S}(t) := \frac{1}{||\p{\tu}{x} \, M(t)||^2_{L^2(0, 2\pi)}  } \int_{0}^{2\pi} \left[ \p{\tu}{x} \left[ M(w(t, x)) - M(\tu(t, x))\right]\right]^2 \ dx
\label{eq:S}
\end{equation}
and show its dependence on time $t \in [0,2T]$ for different models in
Figure \ref{fig:S}. We note in this figure that the normalized SGS
stress error $\mathcal{S}(t)$ tends to be quite large in all cases,
which is a known property of Smagorinsky-type models
\cite{GamaharaHattori2017}, and is smallest in the case of Das' \&
Moser's approach \cite{das2002optimal}.

\begin{figure}
\centering
  \includegraphics[scale=0.44]{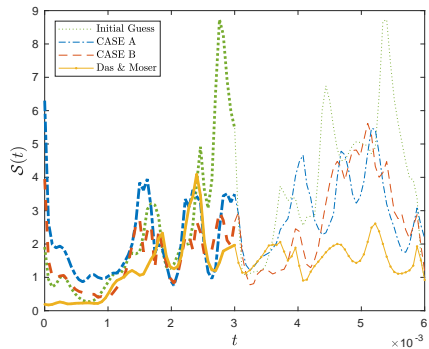}
  \caption{Dependence of the normalized SGS stress error
    $\mathcal{S}(t)$, cf.~\eqref{eq:S}, on time $t$ for the LES
    solutions with (green dotted line) the Smagorinsky model
    \eqref{eq:nusKS}, with (blue dash-dotted line) the optimal eddy
    viscosity $\nud$ from case A, with (red dashed line) the optimal
    eddy viscosity $\nud$ from case B, and with (yellow solid-dotted
    line) the optimal closure proposed by Das \& Moser
    \cite{das2002optimal}. {Bold and thin lines correspond to
      time in the ``training window'' $t \in [0, T]$ and beyond the
      ``training window'' $t \in (T, 2T]$, respectively.}}
\label{fig:S}
\end{figure}

\section{Discussion and Conclusions} 
\label{sec:final}

In this study we have introduced a computational framework for
determining optimal {eddy-viscosity} closures for a broad class of PDE
models, such as the Navier-Stokes equation in hydrodynamics.  This
inverse problem is framed here as a PDE optimization problem where an
error functional representing the misfit between the target and
predicted observations is minimized with respect to the {functional
  form of the eddy viscosity in the closure relation} which determines
the ``shape'' of the nonlinearity of the model term.  Because of this
latter aspect, such a problem is not amenable to solution using
standard adjoint-based tools for PDE optimization and {an
  extension of the recently developed generalization} of these
techniques \cite{bukshtynov2013optimal,bukshtynov2011optimal} needs to
be used.  In addition, by formulating the problem in the
``optimize-then-discretize'' setting we are able to determine the
optimal {forms of the eddy viscosity in a very general manner} subject
only to some minimal assumptions on smoothness,
cf.~\eqref{eq:nue_reg}, and the behavior for small and large values of
the state variable.  {Solution of this problem using a
  ``discretize-then-optimize'' approach often employed to solve PDE
  optimization problems \cite{g03} is an interesting alternative. Such
  formulations} ought to be contrasted with some earlier approaches in
which optimal closures were determined by fitting a small number of
parameters in an assumed ansatz. Thus, our proposed approach does not
suffer from the limitations of such an assumed ansatz.

To provide a proof of the concept for this approach, we have focused
here on the 1D Kuramoto-Sivashinsky equation \eqref{KS} as a model
problem computationally more tractable than the 2D or 3D Navier-Stokes
system we are ultimately interested in. In being chaotic and
multiscale \cite{MR2920624}, the solutions of the Kuramoto-Sivashinsky
system arguably better resemble actual hydrodynamic turbulence than
solutions of the Burgers equation often used as a simplified model in
similar situations \cite{das2002optimal}. Such simplified setting
allows us to study the properties of our proposed approach more
thoroughly. We find that it is possible to determine a particular
dependence of the eddy viscosity on the resolved strain, $\nud =
\nud(|s|)$, generalizing Smagorinsky's relation \eqref{eq:nusKS} in a
LES model, such that the LES model provides a systematically
{improved} approximation of the DNS.  More precisely, while the
trajectories corresponding to the DNS and LES still diverge
exponentially, they do so at a much slower rate than in the case of
the LES based on the standard Smagorinsky model, cf.~Figure
\ref{fig:diagn}(a). Importantly, the functional form of the optimal
eddy viscosity $\nud(|s|)$ appears to little depend on the particular
form of the observations used to set up the optimization problem,
cf.~Figure \ref{fig:nud}, and we emphasize here that our approach can
be formulated based on very general measurements (changing the
observation operator $H_i$, cf.~Section \ref{sec:H}, will only result
in a modification of the source term in the adjoint system
\eqref{eq:aKSLESa}). In particular, it is also possible to
simultaneously use several different sets of measurements coming, for
example, from different DNS or experiments, and a natural way to
formulate such a problem is in terms of multiobjective optimization.

The LES with the optimal {eddy-viscosity} closure models
determined with the proposed approach produces solutions which better
match the DNS than the LES solutions obtained with the closure model
of Das \& Moser \cite{das2002optimal}, cf.~Figure \ref{fig:diagn}(b).
On the other hand, that latter model leads to a more accurate
prediction of the SGS stresses, cf.~Figure \ref{fig:S}. This can be
understood by recognizing that the closure model of Das \& Moser is
formulated to optimally reconstruct the SGS stresses rather than some
other a posteriori quantities. We have also considered {a formulation}
with observation operators involving SGS stresses, but the results
obtained were inferior to the results presented here. {In our
  computations the numerical parameters $N_x$, $N_s$ and $\Delta t$
  were chosen such that the LES and the adjoint systems
  \eqref{eq:KSLES} and \eqref{eq:aKSLES}, the gradient expression
  \eqref{eq:gradL2} as well as the system
  \eqref{eq:gradH3bc}--\eqref{eq:gradH3} for determining the Sobolev
  gradients were fully resolved, cf.~Section \ref{sec:disc}. The effect
  of insufficient numerical resolution on the computed optimal
  eddy-viscosity closures is an important question which merits
  investigation, however, the answer will likely be problem
  dependent.}

In addition to providing optimal {eddy-viscosity} closure models
which can be useful in many situations, the present approach serves
another, more {basic} purpose, namely, by identifying the ``best''
closure models {within a certain family} it can offer information
about {their} fundamental performance limitations. {More
  specifically,} it can provide insights about how well closure models
based on the eddy-viscosity ansatz \eqref{eq:Sm} can perform in the
best case and thus how much room there is in principle for improvement
of standard approaches such as the Smagorinsky model \eqref{eq:nusKS}.

Moving forward, our {next} main goal is to consider an analogous
problem of finding optimal {eddy-viscosity} closure models for
the 2D and 3D Navier-Stokes system as generalizations of the
Smagorinsky model \eqref{eq:Sm}, first in the periodic setting and
then in more realistic geometries. {In that latter context related
  questions also arise as regards wall models.}  We are also
interested in studying closure models based on formulations other than
eddy viscosity.


\section*{Acknowledgments}

The authors acknowledge partial funding through an NSERC (Canada)
Discovery Grant. This work was supported in part by the Research and
High Performance Computing Support (RHPCS) group at McMaster
University.

\appendix 

\section{Proof of Theorem \ref{thm:gradL2}} 
\label{sec:grad}

Here we derive expression \eqref{eq:gradL2} for the $L^2$ gradient of
the cost functional \eqref{eq:J}. In order to avoid {technical}
complications related to the non-differentiability of the absolute
value $|\cdot|$ in the argument of the eddy viscosity, we change the
variable from $|s|$ to $\sigma := |s|^2 = s^2$ (for simplicity and
with a slight abuse of notations we will still use {the} same
symbol $\nu = \nu(\sigma)$ to denote the eddy viscosity).  First, we
must determine the perturbation of the LES system \eqref{eq:KSLES}
resulting from perturbing the functional form of the eddy viscosity
$\nu$ with some perturbation $\nu'$. This is done by replacing $\tu$
and $\nu$ with the following representations in which $0 < \epsilon
\ll 1$
\begin{subequations}
\label{eq:pert}
\begin{align} 
\tu \  & \ \longleftarrow \tu + \epsilon \, \tu' + \mathcal{O}(\epsilon^2), \label{eq:pert_w} \\
\nu(\sigma) \  & \ \longleftarrow  \nu(\sigma) + \epsilon \, \frac{d \nu(\sigma)}{d\sigma} \frac{d\sigma}{ds} \p{\tu'}{x} + \epsilon \, \nu'(\sigma) + \mathcal{O}(\epsilon^2).  \label{eq:pert_nu}
\end{align}
\end{subequations}
The second term on the RHS in \eqref{eq:pert_nu} reflects the change
of the value of the (squared) strain $\sigma$ for which the eddy
viscosity is evaluated as a result of perturbing its functional form
\cite{bukshtynov2013optimal,bukshtynov2011optimal}. Substituting
representations \eqref{eq:pert_w}--\eqref{eq:pert_nu} into the LES
system \eqref{eq:KSLES} and collecting terms of
$\mathcal{O}(\epsilon)$ we obtain the following perturbation system
\begin{subequations}
\label{eq:dKSLES}
\begin{align} 
\p{\tu'}{t} + \nu_4 \, \pppp{\tu'}{x} + \nu_2 \Big[\pp{\tu'}{x} + \p{{(\tu \tu')}}{x}\Big] & 
+ \p{}{x} \Big[ \frac{d\nu(\sigma)}{d\sigma} \frac{d\sigma}{ds} \ppp{\tu}{x} \p{\tu'}{x} + \nu \, \ppp{\tu'}{x} \Big] = -\p{}{x} \Big[\nu' \ppp{\tu}{x} \Big], \label{eq:dKSLESa} \\
\frac{\partial^{(i)} \tu'}{\partial x^{(i)}}(t, 0) &= \frac{\partial^{(i)} \tu'}{\partial x^{(i)}}(t, 2\pi), \quad \quad i = 0, \dotsc, 3,  \label{eq:dKSLESb} \\
\tu'(0, x) &= 0 \label{eq:dKSLESc}
\end{align}
\end{subequations}
describing the leading-order effect $\tu'$ of perturbing the functional
form of the eddy viscosity $\nu(\sigma)$ on solutions of the LES
system \eqref{eq:KSLES}
\cite{bukshtynov2013optimal,bukshtynov2011optimal}.  Now we integrate
\eqref{eq:dKSLESa} against the adjoint field $\tu^*$ over
the space-time domain $[0,2\pi] \times [0,T]$ and then perform
integration by parts with respect to both space and time to obtain
\begin{align*} 
&0 = \int_0^T \int_0^{2\pi} \bigg[\p{\tu'}{t} + \nu_4 \, \pppp{\tu'}{x} + \nu_2 \Big[\pp{\tu'}{x} + \p{{(\tu \, \tu')}}{x}\Big] \bigg] \, \tu^* \, dx \, dt \nonumber \\ & \ \  
+ \int_0^T \int_0^{2\pi} \p{}{x} \bigg[\frac{d\nu}{d\sigma} \, \frac{d\sigma}{ds} \, \ppp{\tu}{x} \, \p{\tu'}{x} + \nu \, \ppp{\tu'}{x} + \nu' \, \ppp{\tu}{x} \ \bigg] \, \tu^* \, dx \, dt \\
= &  \underbrace{\int_0^T \int_0^{2\pi} \left[-\p{\tu^*}{t} + \nu_4 \pppp{\tu^*}{x} + \nu_2 \left[\pp{\tu^*}{x} - \tu\p{\tu^*}{x}\right] 
 + \p{}{x}\left[\frac{d\nu}{d\sigma} \frac{d\sigma}{ds} \ppp{\tu}{x} \p{\tu^*}{x}\right] + \ppp{}{x} \left[\nu \p{\tu^*}{x} \right] \right] \, \tu' \, dx \, dt}_{\J'(\nu;\nu')} \nonumber \\
& -\int_0^T \int_0^{2\pi} \p{\tu^*}{x} \, \ppp{\tu}{x} \, \nu' \, dx \, dt = 0,  \nonumber
\end{align*}
where all the boundary terms vanish due to periodicity. Using the
definition of the adjoint system \eqref{eq:aKSLES} together with the
aforementioned change of variables $\sigma = \sigma(s)$ we then obtain
for the G\^ateaux differential $\J'(\nu;\nu') = \int_0^T \int_0^{2\pi}
\p{\tu^*}{x} \, \ppp{\tu}{x} \, \nu' \, dx \, dt$, which now contains
the perturbation $\nu'$ as a factor, but is still not in the Riesz
form \eqref{eq:Riesz} because this form involves an inner product
defined with integration with respect to the resolved strain $s$
rather than space $x$ and time $t$ (we shall now return {back} to the
original variable via the substitution $\sigma = s^2$). The required
change of variables is introduced by the representation
\cite{bukshtynov2013optimal,bukshtynov2011optimal}
{$\nu'\Big(\Big|\frac{\partial \tu(t, x)}{\partial x}\Big|\Big) =
  \int_a^b \delta\Big(\Big|\frac{\partial \tu(t, x)}{\partial x}\Big|
  - s\Big) \, \nu'(s) \, ds$,} where $\delta(\cdot)$ is the Dirac
delta distribution, such that using Fubini's theorem to swap the order
of integration we finally arrive at a Riesz representation of the
G\^ateaux differential \eqref{eq:dJ}
\begin{align}
\J'(\nu; \nu') &= \int_0^T \int_0^{2\pi} \p{\tu^*(t, x)}{x} \, \int_{a}^{b} \delta\Big(\Big|\frac{\partial \tu(t, x)}{\partial x}\Big| - s\Big) \, \ppp{\tu(t, x)}{x} \, \nu'( s ) \, ds \, dx \, dt \nonumber \\
&= \int_{a}^{b} \left[\int_{0}^T \int_{0}^{2\pi} \p{\tu^*(t, x)}{x} \, \delta\Big(\Big|\frac{\partial \tu(t, x)}{\partial x}\Big| - s\Big) \, \ppp{\tu(t, x)}{x} \, dx \, dt\right] \nu'( s ) \, ds,
\label{eq:dJ2}
\end{align}
from which after selecting $\X = L^2(\LL)$ in \eqref{eq:Riesz} we
deduce the following expression for the $L^2$ gradient
\begin{equation} 
\label{eq:gradL2sig}
\grad_{\nu}^{L^2}\mathcal{J}(s) 
= \int_{0}^T \int_{0}^{2\pi} \p{\tu^*(t, x)}{x} \, \delta\Big(\Big|\frac{\partial \tu(t, x)}{\partial x}\Big| - s\Big) \, \ppp{\tu(t, x)}{x} \, dx \, dt.
\end{equation}
We note that evaluation of this expression for a given value of $s$
requires computation of an integral defined on the level set
$\Big|\frac{\partial \tu(t, x)}{\partial x}\Big| = s$ in the
space-time domain $[0,2\pi] \times [0,T]$ which is rather difficult. A
computationally more convenient approach is obtained using the
following identity (in which the differentiation is understood in the
distributional sense) {$\delta\Big( \Big|\frac{\partial \tu(t,
    x)}{\partial x}\Big| - s \Big) = -
  \frac{d}{ds}{\Xi}_{\left[\alpha, \big|\partial \tu(t, x) /
      \partial x\big|\right]}(s)$}, such that \eqref{eq:gradL2sig}
becomes
\begin{equation} 
\label{eq:gradL2sig2}
\grad_{\nu}^{L^2}\J(s) = 
- \frac{d}{ds}\int_{0}^T \int_{0}^{2\pi} {\Xi}_{\left[\alpha, \big|\frac{\partial \tu(t, x)}{\partial x}\big| \right]}(s) \, \p{\tu^*(t, x)}{x} \, \ppp{\tu(t, x)}{x} \, dx \, dt,
\end{equation}
and expression \eqref{eq:gradL2} is finally obtained.




\begin{thebibliography}{10}

\bibitem{af05}
{\sc R.~A. Adams and J.~F. Fournier}, {\em Sobolev Spaces}, vol.~140 of Pure
  and Applied Mathematics (Amsterdam), Elsevier/Academic Press, Amsterdam,
  2nd~ed., 2003.

\bibitem{bk07}
{\sc J.~Bec and K.~Khanin}, {\em Burgers turbulence}, Phys. Rep., 447 (2007),
  pp.~1--66.

\bibitem{b77}
{\sc M.~S. Berger}, {\em Nonlinearity and Functional Analysis}, Academic Press,
  1977.

\bibitem{bukshtynov2013optimal}
{\sc V.~Bukshtynov and B.~Protas}, {\em Optimal reconstruction of material
  properties in complex multiphysics phenomena}, J. Comput. Phys., 242 (2013),
  pp.~889--914.

\bibitem{bukshtynov2011optimal}
{\sc V.~Bukshtynov, O.~Volkov, and B.~Protas}, {\em On optimal reconstruction
  of constitutive relations}, Phys. D, 240 (2011), pp.~1228--1244.

\bibitem{Canuto1993book}
{\sc C.~Canuto, M.~Y. Hussaini, A.~Quarteroni, and T.~A. Zang}, {\em Spectral
  Methods in Fluid Dynamics}, Springer-Verlag, Berlin, 1988.

\bibitem{cox2002exponential}
{\sc S.~M. Cox and P.~C. Matthews}, {\em Exponential time differencing for
  stiff systems}, J. Comput. Phys., 176 (2002), pp.~430--455.

\bibitem{dp15a}
{\sc I.~Danaila and B.~Protas}, {\em Optimal reconstruction of inviscid
  vortices}, Proceedings of the Royal Society A, 471 (2015), 20150323.

\bibitem{das2002optimal}
{\sc A.~Das and R.~D. Moser}, {\em Optimal large-eddy simulation of forced
  {B}urgers equation}, Phys. Fluids, 14 (2002), pp.~4344--4351.

\bibitem{davidson2015turbulence}
{\sc P.~A. Davidson}, {\em Turbulence: An introduction for scientists and
  engineers}, Oxford University Press, Oxford, 2nd~ed., 2015.

\bibitem{driscoll2014chebfun}
{\sc T.~A. Driscoll, N.~Hale, and L.~N. Trefethen}, {\em {Chebfun Guide}},
  Pafnuty Publications, Oxford, UK, 2014.

\bibitem{duraisamy2018turbulence}
{\sc K.~Duraisamy, G.~Iaccarino, and H.~Xiao}, {\em {Turbulence Modeling in the
  Age of Data}}, Annu. Rev. Fluid Mech., 51 (2019), pp.~357--377.

\bibitem{Durbin2018}
{\sc P.~A. Durbin}, {\em {Some Recent Developments in Turbulence Closure
  Modeling}}, Annu. Rev. Fluid Mech., 50 (2018), pp.~77--103.

\bibitem{frisch1995turbulence}
{\sc U.~Frisch}, {\em Turbulence}, Cambridge University Press, Cambridge, 1995.
\newblock {The legacy of A. N. Kolmogorov}.

\bibitem{GamaharaHattori2017}
{\sc M.~Gamahara and Y.~Hattori}, {\em Searching for turbulence models by
  artificial neural network}, Phys. Rev. Fluids, 2 (2017), 054604.

\bibitem{germano_1992}
{\sc M.~Germano}, {\em Turbulence: the filtering approach}, J. Fluid Mech., 238
  (1992), pp.~325--336.

\bibitem{g03}
{\sc M.~D. Gunzburger}, {\em Perspectives in Flow Control and Optimization},
  SIAM, 2003.

\bibitem{MR2920624}
{\sc P.~Holmes, J.~L. Lumley, G.~Berkooz, and C.~W. Rowley}, {\em {Turbulence,
  Coherent Structures, Dynamical Systems and Symmetry}}, Cambridge Monogr.
  Mech., Cambridge University Press, Cambridge, 2nd~ed., 2012.

\bibitem{jimenez_2018}
{\sc J.~Jimenez}, {\em Machine-aided turbulence theory}, J. Fluid Mech., 854
  (2018), R1.

\bibitem{kassam2005fourth}
{\sc A.~Kassam and L.~N. Trefethen}, {\em {Fourth-Order Time-Stepping for Stiff
  PDEs}}, SIAM J. Sci. Comput., 26 (2005), pp.~1214--1233.

\bibitem{kuramoto1978diffusion}
{\sc Y.~{Kuramoto}}, {\em {Diffusion-Induced Chaos in Reaction Systems}},
  Progress of Theoretical Physics Supplement, 64 (1978), pp.~346--367.

\bibitem{kutz_2017}
{\sc J.~N. Kutz}, {\em Deep learning in fluid dynamics}, J. Fluid Mech., 814
  (2017), pp.~1--4.

\bibitem{langford1999optimal}
{\sc J.~A. Langford and R.~D. Moser}, {\em Optimal {LES} formulations for
  isotropic turbulence}, J. Fluid Mech., 398 (1999), pp.~321--346.

\bibitem{Lesieur1996}
{\sc M.~Lesieur and O.~Metais}, {\em {New Trends in Large-Eddy Simulations of
  Turbulence}}, Annu. Rev. Fluid Mech., 28 (1996), pp.~45--82.

\bibitem{lilly1992proposed}
{\sc D.~K. Lilly}, {\em {A proposed modification of the Germano subgrid-scale
  closure method}}, Physics of Fluids A: Fluid Dynamics, 4 (1992),
  pp.~633--635.

\bibitem{ling_kurzawski_templeton_2016}
{\sc J.~Ling, A.~Kurzawski, and J.~Templeton}, {\em Reynolds averaged
  turbulence modelling using deep neural networks with embedded invariance}, J.
  Fluid Mech., 807 (2016), pp.~155--166.

\bibitem{control:lions1}
{\sc J.~Lions}, {\em Contr{\^{o}}le Optimal de Syst{\`{e}}mes Gouvern{\'{e}}s
  par des {\'{E}}quations aux D{\'{e}}riv{\'{e}}es Partielles}, Dunod, Paris,
  1968.
\newblock {(English translation, Springer-Verlag, New-York (1971))}.

\bibitem{l69}
{\sc D.~Luenberger}, {\em Optimization by Vector Space Methods}, John Wiley and
  Sons, 1969.

\bibitem{Maulik2018}
{\sc R.~Maulik, O.~San, A.~Rasheed, and P.~Vedula}, {\em Data-driven
  deconvolution for large eddy simulations of kraichnan turbulence}, Physics of
  Fluids, 30 (2018), 125109.

\bibitem{Meneveau1994}
{\sc C.~Meneveau}, {\em Statistics of turbulence subgrid‐scale stresses:
  Necessary conditions and experimental tests}, Physics of Fluids, 6 (1994),
  pp.~815--833.

\bibitem{Neuberger2010book}
{\sc J.~Neuberger}, {\em Sobolev Gradients and Differential Equations},
  Springer, 2010.

\bibitem{nw00}
{\sc J.~Nocedal and S.~Wright}, {\em Numerical Optimization}, Springer Series
  in Operations Research and Financial Engineering, Springer, 2nd~ed., 2006.

\bibitem{PanDuraisamy2018}
{\sc S.~Pan and K.~Duraisamy}, {\em Data-driven discovery of closure models},
  SIAM J. Appl. Dyn. Syst., 17 (2018), pp.~2381--2413.

\bibitem{PARISH2016758}
{\sc E.~J. Parish and K.~Duraisamy}, {\em A paradigm for data-driven predictive
  modeling using field inversion and machine learning}, J. Comput. Phys., 305
  (2016), pp.~758 -- 774.

\bibitem{pope2000turbulent}
{\sc S.~B. Pope}, {\em Turbulent flows}, Cambridge University Press, Cambridge,
  2000.

\bibitem{pftv86}
{\sc W.~H. Press, S.~A. Teukolsky, W.~T. Vetterling, and B.~P. Flannery}, {\em
  Numerical Recipes 3rd Edition: The Art of Scientific Computations}, Cambridge
  University Press, 2007.

\bibitem{pbh04}
{\sc B.~Protas, T.~R. Bewley, and G.~Hagen}, {\em {A computational framework
  for the regularization of adjoint analysis in multiscale PDE systems}},
  Journal of Computational Physics, 195 (2004), pp.~49 -- 89.

\bibitem{pnm14}
{\sc B.~Protas, B.~R. Noack, and M.~Morzy{\'n}ski}, {\em {An Optimal Model
  Identification For Oscillatory Dynamics With a Stable Limit Cycle}}, J.
  Nonlin. Sci., 24 (2014), pp.~245--275.

\bibitem{protas_noack_osth_2015}
{\sc B.~Protas, B.~R. Noack, and J.~\"Osth}, {\em {Optimal nonlinear eddy
  viscosity in Galerkin models of turbulent flows}}, J. Fluid Mech., 766
  (2015), pp.~337--367.

\bibitem{rodi2013large}
{\sc W.~Rodi, G.~Constantinescu, and T.~Stoesser}, {\em {Large-Eddy Simulation
  in Hydraulics}}, {CRC Press}, 2013.

\bibitem{Sagaut2006}
{\sc P.~Sagaut}, {\em Large Eddy Simulation for Incompressible Flow}, Sci.
  Comput., Springer-Verlag, Berlin, 3rd~ed., 2006.

\bibitem{sethurajan2015accurate}
{\sc A.~K. Sethurajan, S.~A. Krachkovskiy, I.~C. Halalay, G.~R. Goward, and
  B.~Protas}, {\em {Accurate Characterization of Ion Transport Properties in
  Binary Symmetric Electrolytes Using In Situ NMR Imaging and Inverse
  Modeling}}, The Journal of Physical Chemistry B, 119 (2015),
  pp.~12238--12248.

\bibitem{sivashinsky1988nonlinear}
{\sc G.~I. Sivashinsky}, {\em Nonlinear analysis of hydrodynamic instability in
  laminar flames. {I}. {D}erivation of basic equations}, Acta Astronaut., 4
  (1977), pp.~1177--1206.

\bibitem{smagorinsky1963general}
{\sc J.~Smagorinsky}, {\em {General circulation experiments with the primitive
  equations: I. The basic experiment}}, Monthly weather review, 91 (1963),
  pp.~99--164.

\bibitem{trefethen2000spectral}
{\sc L.~N. Trefethen}, {\em {Spectral Methods in MATLAB}}, SIAM, Philadelphia,
  2000.

\bibitem{trefethen2013approximation}
{\sc L.~N. Trefethen}, {\em {Approximation Theory and Approximation Practice}},
  SIAM, Philadelphia, 2013.

\bibitem{Iliescu2018ROM}
{\sc X.~Xie, M.~Mohebujjaman, L.~Rebholz, and T.~Iliescu}, {\em Data-driven
  filtered reduced order modeling of fluid flows}, SIAM J. Sci. Comput., 40
  (2018), pp.~B834--B857.

\end{thebibliography}

\end{document}